\documentclass[10pt,journal]{IEEEtran}
\IEEEoverridecommandlockouts

\usepackage{graphicx}
\usepackage{amsmath,amsthm,amsfonts,amssymb}
\usepackage{cite}
\usepackage{bm}
\usepackage{bbm}
\usepackage{url}
\usepackage{array}
\usepackage{color}
\usepackage{multirow}
\usepackage{booktabs}
\usepackage[table,xcdraw]{xcolor}
\usepackage{enumitem}
\usepackage{subcaption}

\usepackage[english]{babel}
\usepackage{algorithm}
\usepackage[noend]{algpseudocode}
\usepackage{setspace,soul}

\theoremstyle{plain}
\newtheorem{thm}{Theorem}
\newtheorem{lem}[thm]{Lemma}

\newtheorem{prop}[thm]{Proposition}
\newtheorem{cor}[thm]{Corollary}
\newtheorem{rem}{Remark}

\usepackage[english]{babel}
\usepackage{algorithm}
\usepackage[noend]{algpseudocode}

\allowdisplaybreaks[4]

\renewenvironment{quote}
  {\list{}{\leftmargin=1em \rightmargin=1em}\item\itshape}
  {\endlist}

\begin{document}
\title{Embodied Communication: Sensing-Induced Reliability Fields and Capacity Bounds}

\author{
Yulin~Shao,~\IEEEmembership{Senior Member,~IEEE}
\thanks{The author is with the Department of Electrical and Computer Engineering, The University of Hong Kong, Hong Kong, China (ylshao@hku.hk).}
}

\maketitle

\begin{abstract}
This paper introduces embodied communication, a new wireless communication modality in which information is imprinted onto environmental states and recovered by the receiver through sensing. No dedicated communication transmitter is activated, and no additional communication spectrum is occupied; instead, the sensed environment itself becomes the carrier of information.
The key insight is that sensing must be reinterpreted for communication. Rather than asking how accurately an unknown physical state can be estimated, embodied communication asks how reliably two states can be distinguished. We formalize this idea through a multi-snapshot radio frequency (RF) sensing model and derive a sensing-induced reliability field that quantifies the distinguishability between physical states. This field turns embodied symbol design into a geometric packing problem shaped by the sensing resolution of the infrastructure.
For this embodied channel, we characterize the finite-snapshot $\epsilon$-capacity through achievable designs and converses. We develop lattice-based codebooks, obtain a closed-form hexagonal design under a main-lobe approximation, and establish information-theoretic and geometric upper bounds. 
We further investigate the sensing-duration tradeoff: more sensing snapshots improve reliability, but also lengthen each embodied symbol, leading to a finite optimal sensing time. These results expose a latent communication pathway in sensing-enabled infrastructure and show how the environment can be transformed from a passive backdrop into an active information carrier.
\end{abstract}


\section{Introduction}\label{sec:intro}
\subsection{Infrastructure with Sensing}
The evolution toward 6G wireless networks is marked by a fundamental reorientation of the role of infrastructure \cite{liu2022integrated,liu2022survey}. 
The base station (BS) is no longer envisioned merely as a node that transports data packets. 
Instead, it is increasingly becoming a multimodal sensing platform that perceives, interprets, and interacts with the physical world. 
Integrated sensing and communication (ISAC) \cite{liu2022integrated}, for instance, reuses radio-frequency (RF) spectral and hardware resources to endow the BS with radar-like sensing capability. 
Vision sensors such as cameras \cite{xu2022computer,zhang2024optical} and LiDAR \cite{shao2025point,bian2025over}, increasingly co-located with antenna arrays, further equip the BS with visual and geometric perception.

The sensing capability of wireless infrastructure is typically exploited in two ways.
\begin{itemize}[leftmargin=0.5cm]
    \item \textit{Beamforming assistance}. Environmental awareness is used to enhance the wireless link itself.  A representative application is predictive and site-specific beamforming \cite{liu2020radar,heng2022learning}. By tracking user mobility and reconstructing the surrounding geometry, the BS can preemptively steer beams around obstacles, maintain beam alignment during rapid movement, and reduce the overhead of channel state information (CSI) acquisition. In this mode, sensing forms a closed loop at the BS side: the BS observes the environment and immediately uses the observation to improve its own communication.
    \item \textit{Sensing as a service}. The BS may also operate as part of a distributed sensing network, extracting physical information from the environment and delivering it to downstream applications \cite{boukerche2008vehicular}. Examples include target localization, velocity tracking, 3D scene reconstruction, micro-motion detection, autonomous vehicle coordination, industrial digital twins, and smart city analytics \cite{kotaru2015spotfi,shao2021federated}. In this mode, sensing is provided as a standalone service and is not necessarily coupled to the immediate needs of a communication link.
\end{itemize}

While both modes represent important advances, they share an implicit assumption: the physical environment is a passive object of observation. 
The world is seen, tracked, and measured, but it does not intentionally carry information for the BS to decode.

\begin{figure}[t]
  \centering
  \includegraphics[width=0.9\columnwidth]{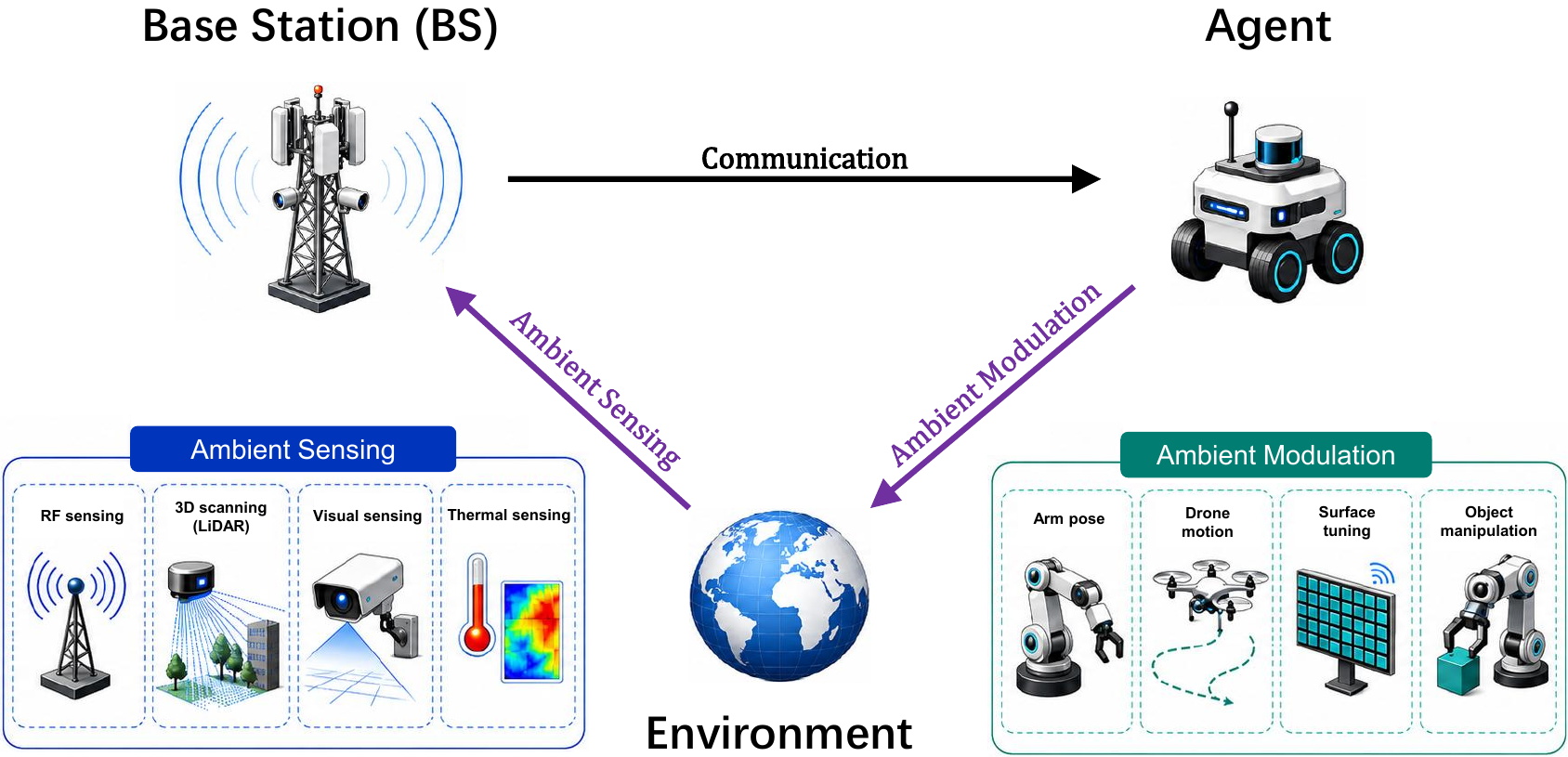}
  \caption{Conceptual illustration of embodied communication. 
Information is embodied in agent-induced environmental states and decoded through BS sensing observations, rather than transmitted by a dedicated uplink waveform.}
  \label{fig:1}
\end{figure}

\subsection{The Vision of Embodied Communication}
Consider the relationships illustrated in Fig.~\ref{fig:1}. 
There are three entities: the BS, the communication agent, and the environment. 
Two links are already familiar. 
The first is the sensing link from the environment to the BS, through which the BS observes physical states. 
The second is the conventional communication link from the BS to the agent, through which information is delivered by a designed waveform.

There is, however, a third link that is easily overlooked. 
It points from the agent to the environment. 
After all, the agent is physically situated in the environment and can deliberately alter environmental states, whether by changing posture, manipulating nearby objects, blocking or revealing a thermal source, or changing the reflective properties of surfaces. 
These actions are not merely physical side effects. 
If selected from an interpretable set of states, they can carry information.

Once this third arrow is acknowledged, a new communication pathway emerges. 
The agent's deliberate manipulation of the environment, followed by the BS's observation of that same environment, forms an information cascade from the agent to the BS through the environment. 
The agent does not generate a dedicated communication waveform. 
Rather, the agent inscribes a change onto the world, and the BS, through its sensing faculties, reads it. 
This constitutes a pathway through which information can flow. 
No additional transmitter is required, and no new spectrum is occupied. 
The pathway is already there, latent in the very architecture of a sensing-enabled network. It simply awaits intentional use.

We refer to this new form of communication as \emph{embodied communication}. Formally,
\begin{quote}
Embodied communication is a communication modality in which information is encoded into deliberately induced environmental states and decoded from the receiver's sensory observations.
\end{quote}

The term ``embodied'' emphasizes that the information is physically instantiated in the state of the environment itself. 
It is inspired by the intellectual lineage of embodied AI and embodied cognition \cite{savva2019habitat,chrisley2003embodied,shao2024theory,shapiro2019embodied}, where intelligence and meaning arise from an agent's interaction with the physical world rather than from disembodied computation alone \cite{chrisley2003embodied}. 
In a similar spirit, embodied communication treats the environment not as a passive backdrop for signal propagation, but as a medium through which expression can be enacted.

To study embodied communication, we need to identify the channel through which information flows. 
This channel is neither a conventional waveform channel nor a sensing channel alone. 
It is formed by the coupling between the agent's physical modulation of the environment and the receiver's sensing of the resulting state. 
We call this end-to-end channel the \emph{embodied channel}.
It consists of two coupled components:
\begin{itemize}[leftmargin=0.5cm]
    \item {Ambient modulation}: the agent-side operation that maps information to a controllable environmental state.
    \item {Ambient sensing}: the receiver-side operation that maps the induced environmental state to noisy sensory observations.
\end{itemize}
Together, they form an information-bearing channel through the environment.

This structure makes the embodied channel fundamentally different from a classical communication channel. 
In a conventional channel, the input alphabet consists of waveforms synthesized by a transmitter and constrained by communication resources such as power, bandwidth, and blocklength \cite{tse2005fundamentals}. 
In an embodied channel, the input alphabet consists of physically realizable environmental states. 
It is constrained jointly by what the agent can realize and what the receiver can distinguish. 
The central design question is therefore: how many physical states can be carved out of a controllable environment and used as reliable embodied symbols?

Answering this question requires characterizing the sensing distinguishability between physical states and establishing its connection to communication reliability. Specifically, we need to quantify the statistical distinguishability between the sensing observations induced by different physical states, translate this distinguishability into decoding reliability guarantees, and use these guarantees to guide physical-state codebook design. Importantly, reliable codebook design must account for global sensing ambiguities across the controllable physical region, which cannot be fully captured by local estimation metrics such as the Fisher information matrix (FIM) or Cramér-Rao bound (CRB) \cite{smith2005statistical}.

\subsection{Contributions}

This paper makes the following advances.

\begin{itemize}[leftmargin=0.5cm]
\item We introduce embodied communication, a new communication modality in which information is encoded into deliberately induced environmental states and recovered through receiver-side sensing. We formulate the resulting embodied channel as a cascade of ambient modulation and ambient sensing, and instantiate this framework in an RF sensing-enabled cellular system where an agent conveys information through controllable scatterer positions.

\item We put forth a Bhattacharyya-distance-based characterization of sensing distinguishability for the embodied channel. Specifically, we derive a closed-form expression for the pairwise Bhattacharyya distance between the sensing distributions induced by different scatterer positions, leading to a sensing-induced reliability field over physical displacements. We further derive the maximum-likelihood decoding rule and establish an explicit upper bound on the codebook-level decoding error probability, yielding a sufficient condition for reliable embodied communication.

\item We formulate reliable physical-state codebook design as a forbidden-region packing problem shaped by the sensing-induced reliability field. We develop lattice-based achievable constructions and obtain a closed-form hexagonal design under a main-lobe approximation. Beyond this local approximation, we investigate codebook optimization using the complete reliability field and demonstrate how accounting for nonlocal sensing ambiguities can improve the design of reliable physical-state alphabets.

\item We establish achievable and converse bounds on the finite-snapshot $\epsilon$-capacity of the embodied channel. The achievable constructions provide lower bounds, while support-constrained information-theoretic and geometric packing arguments yield complementary upper bounds. We further investigate the sensing-duration tradeoff, showing how the number of sensing snapshots balances decoding reliability and symbol duration, and derive an approximate optimal snapshot number under the main-lobe hexagonal design.
\end{itemize}
\section{System Model and Embodied Channel}\label{sec:system}
We consider a sensing-enabled cellular system consisting of one BS and one communication agent, as illustrated in Fig.~\ref{fig:2}. 
The BS is endowed with RF sensing capability, motivated by the integration of sensing functions into 5.5G and 6G infrastructures. 
Although embodied communication can in principle be realized over arbitrary sensing modalities, this paper focuses on RF sensing to develop a concrete communication-theoretic model. 
For ease of reference, the main notation used throughout
the paper is summarized in Table~\ref{tab:notation}.

\begin{table}[!t]
    \centering
    \caption{Summary of Main Notation}
    \label{tab:notation}
    \footnotesize
    \setlength{\tabcolsep}{3pt}
    \renewcommand{\arraystretch}{1.13}
    \begin{tabular}{@{}p{0.37\columnwidth}p{0.59\columnwidth}@{}}
        \toprule
        \textbf{Symbol} & \textbf{Description} \\
        \midrule
        $\mathcal{A}$ & Agent-controllable physical region \\
        $a_y,\,a_z$ & Side lengths of $\mathcal{A}$ along $y$ and $z$ \\
        $D$ & BS-agent-plane distance \\
        $M_y,\,M_z,\,M$ & UPA dimensions; $M=M_yM_z$ antennas \\
        $\mathbf{r}=(y,z)$ & Controllable scatterer position \\
        $\mathbf{a}(\mathbf{r})$ & Normalized receive steering vector \\
        $L$ & Sensing snapshots per embodied symbol \\
        $\rho^2,\,\sigma^2,\,\gamma_0$ & Echo power; noise variance; sensing SNR \\
        $\mathbf{R}(\mathbf{r})$ & Position-dependent sensing covariance \\
        $\mathfrak{C},\,J$ & Position codebook and alphabet size \\
        $\epsilon$ & Target maximum decoding error probability \\
        $\boldsymbol{\Delta}_{ij}$ & Codeword displacement $\mathbf{r}_i-\mathbf{r}_j$ \\
        $\eta(\boldsymbol{\Delta})$ & Squared steering-vector correlation \\
        $B(\boldsymbol{\Delta})$ & Single-snapshot Bhattacharyya field \\
        $\mathcal{F}_J$ & Forbidden displacement region for size $J$ \\
        $\mathbf{G},\,\mathbf{G}_{\rm B}$ & Lattice generator; local reliability metric \\
        $J_\epsilon^{\star}(L)$ & Maximum $\epsilon$-reliable alphabet size \\
        $C_\epsilon(L)$ & Finite-snapshot $\epsilon$-capacity \\
        \bottomrule
    \end{tabular}
\end{table}

\subsection{RF Sensing Model at the BS}\label{sec:bs_sensing}
The BS uses a probing waveform to illuminate the environment and employs a colocated uniform planar array (UPA) to receive the scattered echo. 
As illustrated in Fig.~\ref{fig:2}, the receive UPA lies in the $yz$-plane and is centered at the origin. 
Its broadside direction is aligned with the positive $x$-axis. 
The UPA consists of $M_y \times M_z$ isotropic antenna elements with inter-element spacing $d=\lambda/2$, where $\lambda$ is the carrier wavelength. 
The total number of antenna elements is denoted by $M=M_yM_z$.

The communication agent is located in the far field of the BS. 
It controls a planar physical region of size $a_y\times a_z$, where $a_y$ and $a_z$ denote the extents along the $y$- and $z$-axes, respectively. 
The region is located at distance $D$ from the BS, with its center at $(D,0,0)$, and is parallel to the $yz$-plane. 
The agent-controllable region is therefore
\begin{equation}\label{eq:agent_plane}
\mathcal{A}
=
\left[-\frac{a_y}{2},\frac{a_y}{2}\right]
\times
\left[-\frac{a_z}{2},\frac{a_z}{2}\right].
\end{equation}
Within this region, the agent can deliberately control
the position of a point scatterer, thereby creating the simplest form of a controllable spatial pattern.\footnote{More generally, the agent may manipulate multiple point scatterers or extended reflective structures, leading to richer spatial patterns.}
The resulting environmental state serves as the embodied symbol that carries the intended information. 
In other words, the agent communicates by writing a physical state into the environment, and the BS decodes the intended message by sensing that state.

Let $\mathbf{r} = (y,z) \in \mathcal{A}$ denote the position of the controllable point scatterer on the agent plane, where $(y,z)$ is the offset from the center of the plane. 
Under the far-field and small-region assumptions, this physical coordinate is approximately mapped to the angular coordinates observed at the BS as
\begin{equation}\label{eq:angle_mapping_single}
\theta \approx \frac{y}{D},
\qquad
\phi \approx \frac{z}{D}.
\end{equation} 
Throughout the paper, the scatterer may therefore be equivalently described by its physical coordinate $(y,z)$ on the agent plane or by its angular direction $(\theta,\phi)$ at the BS.

\begin{figure}[t]
  \centering
  \includegraphics[width=1\columnwidth]{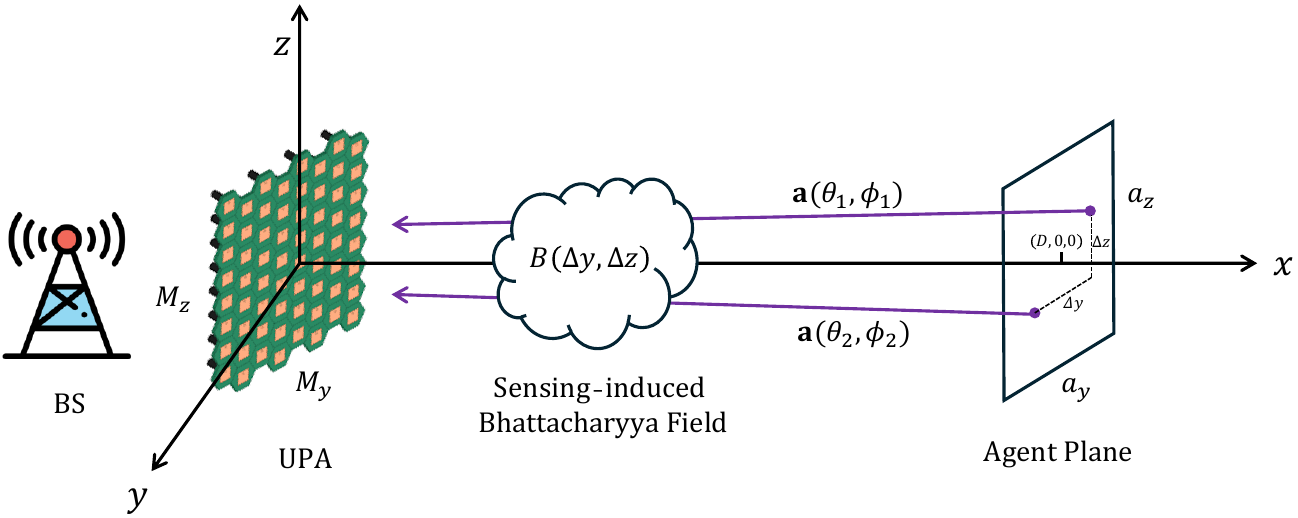}
  \caption{Embodied communication in a sensing-enabled cellular system. The agent communicates by writing a physical state into the environment, and the BS decodes the intended message by sensing the state.}
  \label{fig:2}
\end{figure}

To sense a selected physical state, the BS collects
$L$ sensing snapshots. 
During the $\ell$th snapshot, where $\ell=1,2,\ldots,L$, the BS emits a narrowband probing waveform
\begin{equation}
s_{\ell}(t)=\sqrt{E_t}\bar{s}_{\ell}(t),
\end{equation}
where $E_t$ is the transmit energy per probing snapshot and $\bar{s}_{\ell}(t)$ is a unit-energy baseband pulse with duration $T_p$, e.g., a rectangular pulse $\bar{s}(t) = \frac{1}{\sqrt{T_p}}$ for $0\leq t \leq T_p$.
The probing snapshots occupy separate time intervals.

The probing signal propagates from the BS to the
controllable scatterer, undergoes scattering, and
then propagates back to the BS receive array. 
For a scatterer at position $\mathbf r$, the
effective complex channel coefficient can be represented as
\begin{equation}\label{eq:physical_scattering_single}
h_\ell(\mathbf r)
=
\sqrt{E_tG_{\rm illum}(\mathbf r)}
\frac{
\lambda\sqrt{\sigma_{\rm RCS}}
}{
(4\pi)^{3/2}D_{\mathbf r}^2
}
\xi_\ell
e^{-j4\pi D_{\mathbf r}/\lambda},
\end{equation}
where $G_{\rm illum}(\mathbf r)$ denotes the
transmit illumination gain,
$\sigma_{\rm RCS}$ is the nominal
radar cross section of the scatterer,
$D_{\mathbf r}$ is the BS-scatterer distance,
and $\xi_\ell\sim\mathcal{CN}(0,1)$ is a normalized complex scattering coefficient.
The exponential term accounts for the
round-trip propagation phase, while the receive-array spatial
response is represented separately by the steering
vector introduced below.

We adopt a uniform-illumination sensing model \cite{stoica2007probing,li2007mimo}, a common
assumption in radar sensing when the target is
known to lie within a region of interest but
its exact position is unknown.
In our system, the BS knows the agent-controllable
region $\mathcal A$, but not the particular scatterer
position selected by the agent.
Accordingly, the probing waveform is transmitted
with a broad beam that provides approximately
uniform illumination over the entire region,
and we have
$G_{\rm illum}(\mathbf r)\approx G_{\rm illum}$.

Furthermore, since the agent-controllable region
is small relative to the BS-agent distance,
the variation in propagation distance is negligible,
giving $D_{\mathbf r}\approx D$.
Therefore, the average received scattering
power is given by
\begin{equation}\label{eq:rho}
\rho^2
\triangleq
\mathbb E[|h_\ell(\mathbf r)|^2]
\approx
\frac{
E_tG_{\rm illum}
\lambda^2\sigma_{\rm RCS}
}{
(4\pi)^3D^4
},
\end{equation}
which is consistent with the standard monostatic
radar equation \cite{skolnik1980introduction}.
Substituting \eqref{eq:rho} into \eqref{eq:physical_scattering_single}, the effective
channel coefficients satisfy
$h_\ell\sim\mathcal{CN}(0,\rho^2)$.
The position-dependent round-trip phase does not
affect this distribution because of the circular
symmetry of the scattering coefficients.
For notational simplicity, we henceforth write
$h_\ell$ instead of $h_\ell(\mathbf r)$.

After down-conversion and matched filtering with $\bar{s}_{\ell}(t)$, the BS obtains the complex baseband sensing snapshot
\begin{equation}\label{eq:snapshot_rx_single}
\mathbf{y}_{\ell}
=
h_{\ell}\mathbf{a}(\theta,\phi)
+
\mathbf{n}_{\ell},
\quad
\ell=1,\ldots,L,
\end{equation}
where $\mathbf{y}_{\ell}\in\mathbb{C}^{M}$ and $\mathbf{a}(\theta,\phi)\in\mathbb{C}^{M}$ is the normalized receive steering vector corresponding to the direction of the scatterer. 
The receiver noise satisfies $\mathbf{n}_{\ell}\sim\mathcal{CN}(\mathbf{0},\sigma^2\mathbf{I}_M)$, where $\sigma^2$ denotes the post-matched-filter noise variance. 
The noise vectors $\{\mathbf{n}_{\ell}\}_{\ell=1}^{L}$ are independent across snapshots.

For the UPA, the antenna elements are indexed by $m_y=0,\ldots,M_y-1$ and $m_z=0,\ldots,M_z-1$. 
The normalized receive steering vector can be written as the Kronecker product of two steering vectors:
\begin{equation}\label{eq:steer}
\mathbf{a}(\theta,\phi)
=
\mathbf{a}_y(\theta,\phi)\otimes\mathbf{a}_z(\phi),
\end{equation}
where
\begin{align*}
\mathbf{a}_y(\theta,\phi)
&=
\frac{1}{\sqrt{M_y}}
\left[
1,
e^{j\pi\sin\theta\cos\phi},
\ldots,
e^{j\pi(M_y-1)\sin\theta\cos\phi}
\right]^{\mathsf{T}},
\\
\mathbf{a}_z(\phi)
&=
\frac{1}{\sqrt{M_z}}
\left[
1,
e^{j\pi\sin\phi},
\ldots,
e^{j\pi(M_z-1)\sin\phi}
\right]^{\mathsf{T}}.
\end{align*}
The normalization ensures $\|\mathbf{a}(\theta,\phi)\|=1$.

Using \eqref{eq:angle_mapping_single}, we define the position-induced steering vector
\begin{equation}\label{eq:position_steering}
\mathbf{a}(\mathbf{r})
\triangleq
\mathbf{a}\!\left(\frac{y}{D},\frac{z}{D}\right).
\end{equation}
Then, stacking all $L$ snapshots gives
\begin{equation}\label{eq:multi_snapshot_position}
\mathbf{Y}
=
\mathbf{a}(\mathbf{r})\mathbf{h}^{\mathsf{T}}
+
\mathbf{N},
\end{equation}
where
$\mathbf{Y}
=
[\mathbf{y}_1,\ldots,\mathbf{y}_L]\in\mathbb{C}^{M\times L}$,
$\mathbf{h}
=
[h_1,\ldots,h_L]^{\mathsf{T}}\in\mathbb{C}^{L}$,
and
$\mathbf{N}
=
[\mathbf{n}_1,\ldots,\mathbf{n}_L]\in\mathbb{C}^{M\times L}$.
The average per-snapshot sensing signal-to-noise ratio (SNR) of the controllable scatterer is defined as
\begin{equation}\label{eq:snr_single}
\gamma_0
\triangleq
\frac{\rho^2}{\sigma^2}.
\end{equation}

Since $h_{\ell}\sim\mathcal{CN}(0,\rho^2)$, the conditional distribution of one sensing snapshot given the scatterer position $\mathbf{r}$ is
\begin{equation}\label{eq:single_snapshot_gaussian_single}
\mathbf{y}_{\ell}\mid \mathbf{r}
\sim
\mathcal{CN}
\left(
\mathbf{0},
\mathbf{R}(\mathbf{r})
\right),
\end{equation}
where
\begin{equation*}
\mathbf{R}(\mathbf{r})
=
\rho^2
\mathbf{a}(\mathbf{r})
\mathbf{a}^{H}(\mathbf{r})
+
\sigma^2\mathbf{I}_M
=
\sigma^2
\left(
\mathbf{I}_M
+
\gamma_0
\mathbf{a}(\mathbf{r})
\mathbf{a}^{H}(\mathbf{r})
\right).
\end{equation*}
Thus, the scatterer position is encoded in
the covariance structure of the received snapshots.
Importantly, the covariance matrix retains
the relative phase relationships across
the receive antennas through
$\mathbf a(\mathbf r)\mathbf a^H(\mathbf r)$,
while remaining independent of the unknown
common scattering phase.
This allows the BS to exploit spatial
phase information without requiring
knowledge of the instantaneous phase
of the scattering coefficient.

In practice, the BS observes the sample covariance
\begin{equation}\label{eq:sample_covariance_single}
\widehat{\mathbf{R}}_{\mathbf{y}}
=
\frac{1}{L}
\sum_{\ell=1}^{L}
\mathbf{y}_{\ell}\mathbf{y}_{\ell}^{H}
=
\frac{1}{L}\mathbf{Y}\mathbf{Y}^{H}.
\end{equation}
The number of snapshots $L$ is the sensing resource that can be leveraged in one embodied channel use.

\subsection{Ambient Modulation}\label{sec:ambient_mod}
The agent communicates by selecting one admissible scatterer position from a predefined position codebook. 
Let
\begin{equation*}
\mathfrak{C}
=
\left\{
\mathbf{r}_1,
\mathbf{r}_2,
\ldots,
\mathbf{r}_J
\right\}
\subseteq
\mathcal{A}
\end{equation*}
denote the set of admissible positions, where $\mathbf{r}_{j}
=
(y_j,z_j)$, $j=1,\ldots,J$.
The cardinality $J=|\mathfrak{C}|$ determines the size of the embodied symbol alphabet.

Let $W\in\{1,2,\ldots,J\}$ denote the message to be conveyed in one embodied channel use.
Ambient modulation is defined as the mapping
\begin{equation}\label{eq:ambient_mod_single}
W=j
\quad\mapsto\quad
\mathbf{r}_{j}\in\mathfrak{C}.
\end{equation}
In the considered operating regime, the BS collects sensing snapshots on a much shorter time scale than that of the agent's physical-state transitions. Accordingly, the scatterer position is assumed to remain unchanged during the $L$ sensing snapshots, while it can be reconfigured between successive embodied channel uses.

We assume ideal ambient modulation in this paper. 
That is, for any selected codeword $\mathbf{r}_{j}\in\mathfrak{C}$, the agent can realize the corresponding scatterer position without placement error, unintended scatterers, or temporal drift within the embodied channel use. 
Accordingly, the ambient modulation channel is deterministic: $\Pr(\mathbf{r}=\mathbf{r}_{j}\mid W=j)=1$.

This assumption allows us to isolate the fundamental limitation imposed by the ambient sensing process, namely the ability of the BS to distinguish different physical positions through noisy multi-snapshot RF observations. 
The impact of imperfect ambient modulation, including finite placement accuracy, actuation noise, and parasitic scattering, is left for future investigation.

\subsection{Embodied Channel Law}\label{sec:embodied_channel}

Combining the position modulation model and the multi-snapshot RF sensing model yields the embodied communication channel. 
For a transmitted message $W=j$, the agent places the scatterer at position $\mathbf{r}_j$, and the BS observes $\mathbf{Y}
=
\mathbf{a}(\mathbf{r}_j)\mathbf{h}_j^{\mathsf{T}}
+
\mathbf{N}$.
In practical sensing systems, the instantaneous scattering coefficients are often unknown to the BS due to uncertainties in propagation and scattering. Nevertheless, the scatterer position remains encoded in the spatial covariance of the sensing observations, which preserves the relative phase information across the receive antennas. Therefore, by marginalizing over the unknown scattering coefficients, the embodied channel transition law is
\begin{equation}\label{eq:multi_snapshot_transition_single}
p(\mathbf{Y}\mid W=j)
=
\prod_{\ell=1}^{L}
\frac{1}{\pi^{M}\det\!\left(\mathbf{R}_{j}\right)}
\exp
\left(
-
\mathbf{y}_{\ell}^{H}
\mathbf{R}_{j}^{-1}
\mathbf{y}_{\ell}
\right),
\end{equation}
where 
$\mathbf{R}_{j}
\triangleq
\mathbf{R}(\mathbf{r}_j)
=
\sigma^2
\left(
\mathbf{I}_M
+
\gamma_0
\mathbf{a}(\mathbf{r}_j)
\mathbf{a}^{H}(\mathbf{r}_j)
\right)$.

The information flow can be summarized as
\begin{equation}\label{eq:cascade_single}
W
\to
\mathbf{r}_{W}
\to
\mathbf{R}_{W}
\to
\mathbf{Y}
\to
\widehat{W}.
\end{equation}
The first mapping represents position modulation by the agent, while the second represents RF sensing by the BS. 
Together, they define a discrete-input continuous-output channel whose input symbols are physical positions rather than conventional transmitter-generated waveforms.

\begin{rem}[Alphabet as a design variable] 
For a given codebook, the embodied channel is a discrete-input continuous-output channel, where the receiver decodes which of the prescribed codewords was selected. The distinctive problem in embodied communication, however, lies in constructing the input alphabet itself. Although the agent can select physical states from a continuous controllable region, not every collection of physically admissible states forms a reliable communication alphabet. Whether a set of states can serve as communication symbols depends on the statistical distinguishability of their induced sensing observations. Consequently, the central problem is to determine both the configuration and cardinality of the physical-state codebook under decoding reliability constraints, rather than merely selecting or decoding symbols from a predefined alphabet.
\end{rem}

\section{The $\epsilon$-Capacity of Embodied Channel}\label{sec:capacity}

In this section, we define the $\epsilon$-capacity of the embodied channel introduced. 
The definition follows a finite-blocklength reliability \cite{polyanskiy2010channel} viewpoint: for a fixed number of sensing snapshots $L$ and a target error probability $\epsilon$, we ask \emph{how many physical scatterer positions can be reliably used as embodied symbols}. 
This differs from the Shannon capacity of a fixed input alphabet. 
Here, the input alphabet itself is a design object, since the codewords are physical positions selected from the agent-controllable region $\mathcal{A}$.

\subsection{The $\epsilon$-Capacity}\label{subsec:epsilon_capacity_def}

Consider an embodied codebook $\mathfrak{C}=\{\mathbf{r}_1,\ldots,\mathbf{r}_J\}
\subseteq \mathcal{A}$,
where message $W=j$ is mapped to position $\mathbf{r}_j$. 
Given the observation $\mathbf{Y}$, the BS applies a decoder
\begin{equation*}
g:\mathbb{C}^{M\times L}\rightarrow \{1,\ldots,J\},
\end{equation*}
and outputs $\widehat{W}=g(\mathbf{Y})$.

For a fixed codebook $\mathfrak{C}$ and decoder $g$, we define the maximum decoding error probability as
\begin{equation*}
P_{\max}(\mathfrak{C},g)
=
\max_{j\in\{1,\ldots,J\}}
\Pr\left(g(\mathbf{Y})\neq j\mid W=j\right).
\end{equation*}
The optimal maximum error probability of the codebook is
\begin{equation}\label{eq:pmax_star_def}
P_{\max}^{\star}(\mathfrak{C})
=
\inf_{g}
P_{\max}(\mathfrak{C},g),
\end{equation}
where the infimum is over all measurable decoders.

For a target error probability $\epsilon\in(0,1)$ and a fixed number of snapshots $L$, define the maximum $\epsilon$-reliable embodied alphabet size as
\begin{eqnarray*}
&&\hspace{-0.8cm} J_{\epsilon}^{\star}(L)
=
\max_{\mathfrak{C}\subseteq\mathcal{A}}
|\mathfrak{C}| \\
&&\hspace{-0.8cm} \text{s.t.}~~ P_{\max}^{\star}(\mathfrak{C}) 
\leq
\epsilon. \label{eq:J_epsilon_constraint}
\end{eqnarray*}
The $\epsilon$-capacity of the embodied channel is then defined as
\begin{equation}\label{eq:C_epsilon_def}
C_{\epsilon}(L)
=
\frac{1}{L T_p}
\log_2
J_{\epsilon}^{\star}(L).
\end{equation}

The $\epsilon$-capacity in \eqref{eq:C_epsilon_def} is the central object of this paper. 
It measures the largest number of physical states that can be carved out of the agent plane and reliably distinguished by the BS sensing system. 
To characterize this quantity, we first study the most elementary problem: distinguishing two candidate scatterer positions.

\subsection{Pairwise Distinguishability}\label{subsec:pairwise_distinguishability}
Consider two candidate positions $\mathbf{r}_i=(y_i,z_i)$, $\mathbf{r}_j=(y_j,z_j)$ with corresponding steering vectors $\mathbf{a}_i=\mathbf{a}(\mathbf{r}_i)$ and $\mathbf{a}_j=\mathbf{a}(\mathbf{r}_j)$.
The BS must distinguish between the two hypotheses:
\begin{align*}
\mathcal{H}_i &: 
\mathbf{y}_{\ell}
\sim
\mathcal{CN}(\mathbf{0},\mathbf{R}_i),\\
\mathcal{H}_j &: 
\mathbf{y}_{\ell}
\sim
\mathcal{CN}(\mathbf{0},\mathbf{R}_j).
\end{align*}

For a single snapshot, the optimal likelihood-ratio test compares the likelihoods induced by $\mathbf{R}_i$ and $\mathbf{R}_j$. 
For $L$ independent snapshots, the log-likelihood ratio is the sum of the corresponding single-snapshot log-likelihood ratios. 
Although the exact pairwise error probability of this covariance discrimination problem generally does not admit a simple closed form, it can be tightly characterized by classical error exponents. 
In particular, we use the Bhattacharyya bound \cite{kailath1967divergence}, which gives an analytically tractable upper bound on the pairwise error probability.

For two zero-mean complex Gaussian distributions with covariance matrices $\mathbf{R}_i$ and $\mathbf{R}_j$, the single-snapshot Bhattacharyya distance is
\begin{equation}\label{eq:Bij_general}
B_{ij}
=
\log
\frac{
\det\left(\frac{\mathbf{R}_i+\mathbf{R}_j}{2}\right)
}{
\sqrt{\det(\mathbf{R}_i)\det(\mathbf{R}_j)}
}.
\end{equation}
Since the $L$ snapshots are conditionally independent, the corresponding pairwise error probability satisfies
\begin{equation}\label{eq:pairwise_bhatt_bound_general}
P_{i\rightarrow j}
\leq
\exp(-L B_{ij}).
\end{equation}

We now evaluate $B_{ij}$ for the embodied sensing model.

\begin{thm}[Pairwise embodied sensing exponent]\label{thm:pairwise_error}
For two scatterer positions $\mathbf{r}_i$ and $\mathbf{r}_j$, define the squared steering-vector correlation $\eta_{ij}
\triangleq
\left|
\mathbf{a}_i^H\mathbf{a}_j
\right|^2$.
Then the Bhattacharyya distance between the two induced sensing distributions is
\begin{equation}\label{eq:Bij_single}
B_{ij}
=
\log
\frac{
\left(1+\frac{\gamma_0}{2}\right)^2
-
\frac{\gamma_0^2}{4}\eta_{ij}
}{
1+\gamma_0
}.
\end{equation}
Consequently, the pairwise error probability from $L$ snapshots satisfies
\begin{equation}\label{eq:pairwise_error_eta}
P_{i\rightarrow j}
\leq
\left[
\frac{
1+\gamma_0
}{
\left(1+\frac{\gamma_0}{2}\right)^2
-
\frac{\gamma_0^2}{4}\eta_{ij}
}
\right]^L.
\end{equation}
\end{thm}

\begin{proof}
Since $\|\mathbf{a}_i\|=\|\mathbf{a}_j\|=1$, the matrix determinant lemma gives $\det(\mathbf{R}_i)=\det(\mathbf{R}_j)
=
\sigma^{2M}(1+\gamma_0)$.
Thus, the position-dependent part of \eqref{eq:Bij_general} is contained in $\det\left(\frac{\mathbf{R}_i+\mathbf{R}_j}{2}\right)$, where
\begin{align*}
\frac{\mathbf{R}_i+\mathbf{R}_j}{2}
&=
\sigma^2
\left[
\mathbf{I}_M
+
\frac{\gamma_0}{2}
\left(
\mathbf{a}_i\mathbf{a}_i^H
+
\mathbf{a}_j\mathbf{a}_j^H
\right)
\right].
\end{align*}

Defining $\mathbf{U}_{ij}
=
[\mathbf{a}_i,\mathbf{a}_j]
\in\mathbb{C}^{M\times 2}$, we have
\begin{align*}
\det\left(\frac{\mathbf{R}_i+\mathbf{R}_j}{2}\right)
&=
\sigma^{2M}
\det
\left[
\mathbf{I}_M
+
\frac{\gamma_0}{2}
\mathbf{U}_{ij}\mathbf{U}_{ij}^H
\right].
\end{align*}
By Sylvester's determinant identity,
\begin{equation*}
\det
\left[
\mathbf{I}_M
+
\frac{\gamma_0}{2}
\mathbf{U}_{ij}\mathbf{U}_{ij}^H
\right]
=
\det
\left[
\mathbf{I}_2
+
\frac{\gamma_0}{2}
\mathbf{U}_{ij}^H\mathbf{U}_{ij}
\right].
\end{equation*}
Moreover,
\begin{equation*}
\mathbf{U}_{ij}^H\mathbf{U}_{ij}
=
\begin{bmatrix}
1 & \mathbf{a}_i^H\mathbf{a}_j\\
\mathbf{a}_j^H\mathbf{a}_i & 1
\end{bmatrix}.
\end{equation*}
Hence,
\begin{align*}
\det
\left[
\mathbf{I}_2
+
\frac{\gamma_0}{2}
\mathbf{U}_{ij}^H\mathbf{U}_{ij}
\right]
&=
\det
\begin{bmatrix}
1+\frac{\gamma_0}{2}
&
\frac{\gamma_0}{2}\mathbf{a}_i^H\mathbf{a}_j\\
\frac{\gamma_0}{2}\mathbf{a}_j^H\mathbf{a}_i
&
1+\frac{\gamma_0}{2}
\end{bmatrix}\\
&=
\left(1+\frac{\gamma_0}{2}\right)^2
-
\frac{\gamma_0^2}{4}\eta_{ij}.
\end{align*}
Therefore,
\begin{equation*}
\det\left(\frac{\mathbf{R}_i+\mathbf{R}_j}{2}\right)
=
\sigma^{2M}
\left[
\left(1+\frac{\gamma_0}{2}\right)^2
-
\frac{\gamma_0^2}{4}\eta_{ij}
\right].
\end{equation*}
Substituting the above determinant expressions into \eqref{eq:Bij_general} yields \eqref{eq:Bij_single}. 
The pairwise error bound \eqref{eq:pairwise_error_eta} follows from \eqref{eq:pairwise_bhatt_bound_general}.
\end{proof}

Theorem~\ref{thm:pairwise_error} shows that the two positions enter the pairwise error exponent only through the squared steering-vector correlation $\eta_{ij}$.
A larger $\eta_{ij}$ means that the two positions induce more similar BS array responses, which reduces $B_{ij}$ and increases the pairwise error probability. 
A smaller $\eta_{ij}$ means that the two positions are more separated in the BS sensing space, which increases $B_{ij}$ and improves discrimination reliability.

% \begin{rem}[From sensing accuracy to sensing distinguishability]\label{rem:sensing_distinguishability}
% The role of sensing in embodied communication is different from that in conventional localization. 
% Classical sensing typically asks how accurately an unknown scatterer position can be estimated. 
% Here, the scatterer position is intentionally selected from a finite physical codebook and therefore acts as a communication symbol. 
% The relevant question is thus how reliably two physical states can be distinguished.

% This is why we characterize the BS sensing process through the binary hypothesis testing distance $B_{ij}$ rather than through an estimation-accuracy metric. 
% The resulting pairwise distinguishability directly controls the probability of confusing one embodied symbol with another, and later becomes the geometry used to design reliable position codebooks.
% \end{rem}

\subsection{Sensing-Induced Reliability Field}
\label{subsec:reliability_field}

Theorem~\ref{thm:pairwise_error} expresses pairwise reliability in terms of $\eta_{ij}$.
We now connect $\eta_{ij}$ to physical displacement on the agent plane.

Let $\boldsymbol{\Delta}_{ij}=(\Delta y_{ij},\Delta z_{ij})\triangleq\mathbf{r}_i-\mathbf{r}_j$.
Under the far-field and small-region approximation in \eqref{eq:angle_mapping_single}, the steering-vector correlation can be expressed in terms of the physical displacement.

\begin{lem}\label{lem:steering_correlation}
Under the UPA model, the squared correlation between two steering vectors associated with positions separated by $\boldsymbol{\Delta}=(\Delta y,\Delta z)$ is
\begin{equation}\label{eq:eta_product}
\eta(\boldsymbol{\Delta})
=
\eta_y(\Delta y)\eta_z(\Delta z),
\end{equation}
where
\begin{equation*}
\eta_y(\Delta y)
\!\!=\!\!
\frac{1}{M_y^2}\!
\left|
\frac{
\sin\left(\frac{\pi M_y \Delta y}{2D}\right)
}{
\sin\left(\frac{\pi \Delta y}{2D}\right)
}
\right|^2  \!\!\!\!,
\eta_z(\Delta z)
\!\!=\!\!
\frac{1}{M_z^2}\!
\left|
\frac{
\sin\left(\frac{\pi M_z \Delta z}{2D}\right)
}{
\sin\left(\frac{\pi \Delta z}{2D}\right)
}
\right|^2.
\end{equation*}
\end{lem}

Lemma~\ref{lem:steering_correlation} is a classical UPA steering-vector correlation result, obtained from the standard Dirichlet-kernel identity under the far-field and small-angle approximation \cite{van2002optimum}. Its role here is to establish the translation-invariant structure of the sensing response correlation: the squared correlation between two physical states depends only on their relative displacement $\boldsymbol{\Delta}$. This property allows the pairwise statistical distinguishability of physical states to be represented as a function of displacement.

Combining Theorem~\ref{thm:pairwise_error} and Lemma~\ref{lem:steering_correlation}, we can define the sensing-induced Bhattacharyya field:
\begin{equation}\label{eq:B_delta_def}
B(\boldsymbol{\Delta})
=
\log
\frac{
\left(1+\frac{\gamma_0}{2}\right)^2
-
\frac{\gamma_0^2}{4}
\eta(\boldsymbol{\Delta})
}{
1+\gamma_0
}.
\end{equation}
Then, the pairwise error probability satisfies
\begin{equation}\label{eq:pairwise_error_Bdelta}
P_{i\rightarrow j}
\leq
\exp\left[-L B(\boldsymbol{\Delta}_{ij})\right].
\end{equation}

The field $B(\boldsymbol{\Delta})$ is the fundamental reliability object induced by the BS sensing system on the agent plane. 
It maps a physical displacement between two possible embodied symbols to a sensing error exponent. 

\begin{rem}[Geometry of $B(\boldsymbol{\Delta})$]
Equations \eqref{eq:eta_product} and \eqref{eq:B_delta_def} reveal that the sensing-induced reliability field inherits the spatial characteristics of the BS receive array. Specifically, the steering-vector correlation depends on the displacement components along the two array axes, rather than solely on the Euclidean distance between two physical states. This gives rise to a generally anisotropic reliability field, where equal displacements along different directions may result in different Bhattacharyya distances.

The physical origin of this anisotropy lies in the directional sensing resolution of the antenna array. For a UPA with unequal apertures along the two axes, a physical displacement along the direction with a larger aperture induces a more pronounced change in the received spatial phase pattern and is therefore more easily distinguished. Consequently, the reliability field typically exhibits elongated contours along the less-resolved direction. Near the origin, these contours can be approximated by ellipses, as further characterized by the main-lobe approximation in Section \ref{subsec:main_lobe_design}. Even for a square UPA, the exact reliability field is not necessarily rotationally invariant due to the underlying array correlation structure.
\end{rem}

For a target pairwise error probability $\epsilon_{\rm p}$, define the required single-snapshot exponent
\begin{equation}\label{eq:B_required_pairwise}
B_{\rm req}(\epsilon_{\rm p},L)
=
\frac{1}{L}
\log\frac{1}{\epsilon_{\rm p}}.
\end{equation}
The pairwise confusability region is
\begin{equation*}
\mathcal{S}_{\rm conf}(\epsilon_{\rm p},L)
=
\left\{
\boldsymbol{\Delta}\in\mathbb{R}^2:
B(\boldsymbol{\Delta})
<
B_{\rm req}(\epsilon_{\rm p},L)
\right\}.
\end{equation*}
If the displacement between two codewords lies outside $\mathcal{S}_{\rm conf}(\epsilon_{\rm p},L)$, then the Bhattacharyya bound guarantees pairwise error probability no larger than $\epsilon_{\rm p}$.

It is worth noting that classical studies often characterize localization performance through estimation-theoretic metrics such as the Fisher information matrix (FIM) and the Cramer-Rao bound (CRB). 

The Bhattacharyya field $B(\boldsymbol{\Delta})$ provides a pairwise error exponent between two candidate embodied symbols and leads directly to codebook-level reliability conditions. 
In contrast, the CRB provides a local lower bound on estimation variance and does not directly yield a maximum-error guarantee for a finite position codebook.

The two viewpoints are nevertheless connected. 
Let $\boldsymbol{\theta}=[y,z]^T$ denote the scatterer position, and let $\mathbf{J}(\boldsymbol{\theta})$ be the single-snapshot FIM associated with the Gaussian model $\mathcal{CN}(\mathbf{0},\mathbf{R}(\boldsymbol{\theta}))$. 
For two nearby positions $\boldsymbol{\theta}$ and $\boldsymbol{\theta}+\boldsymbol{\delta}$, the Bhattacharyya distance admits the local expansion
\begin{equation}
B(\boldsymbol{\theta},\boldsymbol{\theta}+\boldsymbol{\delta})
=
\frac{1}{8}
\boldsymbol{\delta}^{T}
\mathbf{J}(\boldsymbol{\theta})
\boldsymbol{\delta}
+
o(\|\boldsymbol{\delta}\|^2).
\end{equation}
Thus, the FIM underlying the CRB describes the local curvature of the pairwise distinguishability field. 
However, embodied communication requires a global codebook design over finite separations, for which the full function $B(\boldsymbol{\Delta})$ is needed. 
In particular, $B(\boldsymbol{\Delta})$ captures anisotropy, sidelobes, and nonlocal ambiguities of the BS array response that are not represented by a purely local CRB.

\subsection{Reliability of a Codebook}\label{subsec:codebook_reliability}

We now lift the pairwise reliability law to a codebook-level guarantee under maximum error.

For a codebook $\mathfrak{C}$, the maximum-likelihood decoder is
\begin{equation}\label{eq:ML_decoder_general}
\widehat{W}
=
\arg\max_{j\in\{1,\ldots,J\}}
p(\mathbf{Y}|W=j).
\end{equation}
Using \eqref{eq:multi_snapshot_transition_single}, the negative log-likelihood under hypothesis $W=j$ is, up to constants independent of $j$, $L\log\det(\mathbf{R}_j)
+
\operatorname{tr}
\left(
\mathbf{R}_j^{-1}\mathbf{Y}\mathbf{Y}^H
\right)$.
Therefore,
\begin{equation*}\label{eq:ML_decoder_covariance}
\widehat{W}
=
\arg\min_j
\left[
L\log\det(\mathbf{R}_j)
+
\operatorname{tr}
\left(
\mathbf{R}_j^{-1}\mathbf{Y}\mathbf{Y}^H
\right)
\right],
\end{equation*}
where $\det(\mathbf{R}_j)
=
\sigma^{2M}(1+\gamma_0)$, $\forall j$, which is common to all hypotheses.
Substituting
\begin{equation*}
\mathbf{R}_j^{-1}
=
\frac{1}{\sigma^2}
\left(
\mathbf{I}_M
-
\frac{\gamma_0}{1+\gamma_0}
\mathbf{a}_j\mathbf{a}_j^H
\right)
\end{equation*}
and removing terms independent of $j$, we obtain
\begin{equation}\label{eq:ML_decoder_energy}
\widehat{W}
=
\arg\max_j
\mathbf{a}_j^H
\mathbf{Y}\mathbf{Y}^{H}
\mathbf{a}_j.
\end{equation}
Thus, in the embodied channel, the optimal decoder selects the candidate position whose steering vector captures the largest accumulated sensing energy over the $L$ snapshots.

The reliability of the entire codebook is controlled by the weakest pair of codewords. 
Define the minimum pairwise Bhattacharyya distance of $\mathfrak{C}$ as
\begin{equation*}\label{eq:Bmin_def}
B_{\min}(\mathfrak{C})
=
\min_{i\neq j}
B(\mathbf{r}_i-\mathbf{r}_j).
\end{equation*}

\begin{thm}[Sufficient condition for $\epsilon$-reliability]\label{thm:codebook_error}
For a codebook $\mathfrak{C}$ with $J$ codewords, the maximum error probability under ML decoding satisfies
\begin{equation}\label{eq:Pmax_union}
P_{\max}(\mathfrak{C},g_{\rm ML})
\leq
(J-1)
\exp\left[
-LB_{\min}(\mathfrak{C})
\right].
\end{equation}
Consequently, a sufficient condition for $\epsilon$-reliability is
\begin{equation}\label{eq:Bmin_condition}
B_{\min}(\mathfrak{C})
\geq
\frac{1}{L}
\log
\frac{J-1}{\epsilon}.
\end{equation}
\end{thm}

\begin{proof}
The proof follows the union bound over the $J-1$ competing codewords, together with the pairwise bound in \eqref{eq:pairwise_error_Bdelta}, and is omitted for brevity.
\end{proof}

Theorem~\ref{thm:codebook_error} gives a direct codebook design rule. 
A $J$-position codebook is guaranteed to be $\epsilon$-reliable if every pairwise displacement lies outside the confusability region associated with the pairwise budget $\epsilon/(J-1)$. 
This also indicates that a larger codebook requires a larger minimum pairwise error exponent to maintain the same overall reliability.

\section{Achievable Bounds and Codebook Designs}\label{sec:achievable}

Section~\ref{sec:capacity} defines the $\epsilon$-capacity $C_{\epsilon}(L)$ and establishes a sufficient reliability condition for a codebook. 
Specifically, a $J$-position codebook $\mathfrak{C}\subseteq\mathcal{A}$ is guaranteed to be $\epsilon$-reliable if \eqref{eq:Bmin_condition} is satisfied.
In this section, we develop achievable lower bounds on $C_{\epsilon}(L)$ by constructing codebooks that satisfy this condition.

\subsection{Forbidden-Region Packing Lower Bound}\label{subsec:forbidden_region_packing}

The key observation is that the reliability condition defines a forbidden region in displacement space. 
Two physical positions can be used as reliable codewords only if their displacement does not fall into this region. 
Therefore, embodied codebook design becomes a packing problem on the agent plane under a sensing-induced forbidden displacement geometry.

For a candidate codebook size $J$, define the codebook-level Bhattacharyya threshold
\begin{equation*}
B_J
\triangleq
\frac{1}{L}
\log\frac{J-1}{\epsilon}.
\end{equation*}
The corresponding forbidden displacement region is
\begin{equation}\label{eq:FJ_def}
\mathcal{F}_J
=
\left\{
\boldsymbol{\Delta}\in\mathbb{R}^2:
B(\boldsymbol{\Delta})<B_J
\right\}.
\end{equation}
Therefore, a sufficient condition for $\epsilon$-reliability is $\mathbf{r}_i-\mathbf{r}_j
\notin
\mathcal{F}_J$, $\forall i\neq j$.
This leads to the following alphabet size:
\begin{equation}\label{eq:JF_star_def}
J_{\rm F}^{\star}(L,\!\epsilon)
\!=\!
\!\max\!
\left\{\!
J \!:\!
\exists\,
\mathfrak{C}\!\subseteq\!\mathcal{A},
|\mathfrak{C}|\!=\!J,\!
\mathbf{r}_i\!-\!\mathbf{r}_j
\!\notin\!
\mathcal{F}_J,\!
\forall i\!\neq\! j
\right\}.
\end{equation}
Every codebook counted by \eqref{eq:JF_star_def} satisfies the sufficient reliability condition in Theorem~\ref{thm:codebook_error}. 
Hence, the forbidden-region packing rate
\begin{equation}\label{eq:RF_def}
R_{\rm F}(L,\epsilon)
=
\frac{1}{L T_p}
\log_2
J_{\rm F}^{\star}(L,\epsilon)
\end{equation}
is achievable, and $C_{\epsilon}(L)\geq R_{\rm F}(L,\epsilon)$.

The optimization in \eqref{eq:JF_star_def}, however, is difficult because it allows arbitrary finite point sets in a continuous region. 
We next introduce lattice codebooks as a structured and analytically tractable class of embodied codebooks.

\subsubsection{Lattice codebooks}\label{subsec:lattice_codebooks}

A natural structured construction is to select codewords from a two-dimensional lattice on the agent plane. 
Let
\begin{equation*}
\Lambda(\mathbf{G})
=
\left\{
\mathbf{G}\mathbf{k}:
\mathbf{k}\in\mathbb{Z}^2
\right\},
\end{equation*}
where $\mathbf{G}=[\mathbf{g}_1,\mathbf{g}_2]\in\mathbb{R}^{2\times 2}$ is a full-rank lattice generator \cite{zamir2014lattice}. 
The corresponding lattice codebook is
\begin{equation*}\label{eq:lattice_codebook_def}
\mathfrak{C}_{\Lambda}(\mathbf{G})
=
\Lambda(\mathbf{G})\cap\mathcal{A},
\end{equation*}
and we denote by $J_{\Lambda}(\mathbf{G})\triangleq |\mathfrak{C}_{\Lambda}(\mathbf{G})|$ the number of lattice points retained inside the agent plane.

\begin{prop}[Lattice-codebook achievable rate]\label{prop:lattice_rate}
For any full-rank lattice generator $\mathbf{G}$, if
\begin{equation}\label{eq:lattice_avoid_condition}
\mathbf{G}\mathbf{k}
\notin
\mathcal{F}_{J_{\Lambda}(\mathbf{G})},
\quad
\forall \mathbf{k}\in\mathbb{Z}^2\setminus\{\mathbf{0}\},
\end{equation}
the lattice codebook $\mathfrak{C}_{\Lambda}(\mathbf{G})$ is $\epsilon$-reliable under ML decoding. 
The corresponding achievable embodied rate is
\begin{equation}\label{eq:Rlattice_def}
R_{\Lambda}(\mathbf{G})
=
\frac{1}{L T_p}
\log_2
J_{\Lambda}(\mathbf{G}).
\end{equation}
\end{prop}

\begin{proof}
For a lattice codebook, all pairwise displacements belong to the difference lattice: $\mathbf{r}_i-\mathbf{r}_j
\in
\Lambda(\mathbf{G})$, $i\neq j$.

If \eqref{eq:lattice_avoid_condition} holds, then for all distinct codewords $\mathbf{r}_i,\mathbf{r}_j\in\mathfrak{C}_{\Lambda}(\mathbf{G})$, $\mathbf{r}_i-\mathbf{r}_j
\notin
\mathcal{F}_{J_{\Lambda}(\mathbf{G})}$.
By the definition of $\mathcal{F}_{J}$ in \eqref{eq:FJ_def}, this implies
\begin{equation}
B(\mathbf{r}_i-\mathbf{r}_j)
\geq
\frac{1}{L}
\log
\frac{
J_{\Lambda}(\mathbf{G})-1
}{
\epsilon
},
\qquad
\forall i\neq j.
\end{equation}
Therefore, Theorem~\ref{thm:codebook_error} gives $P_{\max}(\mathfrak{C}_{\Lambda}(\mathbf{G}),g_{\rm ML})
\leq
\epsilon$.
The lattice codebook is feasible for the $\epsilon$-capacity problem and achieves \eqref{eq:Rlattice_def}.
\end{proof}

The best lattice achievable rate is
\begin{equation*}
R_{\Lambda}^{\star}(L,\epsilon)
=
\max_{\mathbf{G}\in\mathbb{R}^{2\times 2}}
\frac{1}{L T_p}
\log_2
J_{\Lambda}(\mathbf{G}),~~\text{s.t.}~~\eqref{eq:lattice_avoid_condition}.
\end{equation*}
Then, $R_{\Lambda}^{\star}(L,\epsilon)\leq C_{\epsilon}(L)$ is a lower bound.

\begin{rem}
The lattice formulation contains rectangular grids as a special case by choosing $\mathbf{G}= [[s_y, 0]^\top,[0, s_z]^\top]$.
However, a general lattice is more flexible. 
It can rotate, shear, and adapt its fundamental cell to the anisotropic shape of the forbidden region. 
This additional geometric freedom is essential when $B(\boldsymbol{\Delta})$ is not isotropic.
\end{rem}

For large agent planes, boundary effects are small and the number of retained lattice points is approximately
\begin{equation}\label{eq:lattice_size_approx}
J_{\Lambda}(\mathbf{G})
\approx
\frac{a_y a_z}{|\det\mathbf{G}|}.
\end{equation}
Thus, lattice codebook design can be interpreted as minimizing the fundamental cell area $|\det\mathbf{G}|$ while ensuring that no nonzero lattice point falls into the forbidden region. 
The next subsection shows that, under a main-lobe approximation, this lattice packing problem has a closed-form geometric solution.

\subsubsection{Main-lobe closed-form design}\label{subsec:main_lobe_design}

The exact forbidden region $\mathcal{F}_J$ is induced by the UPA Dirichlet kernels in \eqref{eq:eta_product}. 
It can be anisotropic and may contain sidelobe-induced nonconvex structure. 
However, in many practical regimes, the dominant part of the forbidden region is governed by the main lobe around the origin. 
We now use a second-order approximation to obtain a closed-form lattice design.

\begin{lem}[Quadratic approximation of the reliability field]\label{lem:quadratic_B}
In the main-lobe regime, the sensing-induced Bhattacharyya field admits the approximation
\begin{equation}\label{eq:B_quadratic}
B(\Delta y,\Delta z)
\approx
\kappa
\left(
\alpha_y\Delta y^2+\alpha_z\Delta z^2
\right),
\end{equation}
where $\kappa
=
\frac{\gamma_0^2}{4(1+\gamma_0)}$, and
\begin{equation*}
\alpha_y
=
\frac{\pi^2(M_y^2-1)}{12D^2},
\quad
\alpha_z
=
\frac{\pi^2(M_z^2-1)}{12D^2}.
\end{equation*}
Equivalently, we can write
\begin{equation}\label{eq:B_quadratic_matrix}
B(\boldsymbol{\Delta})
\approx
\boldsymbol{\Delta}^{T}
\mathbf{G}_{\rm B}
\boldsymbol{\Delta},
\end{equation}
where
\begin{equation*}
\mathbf{G}_{\rm B}
=
\kappa
\begin{bmatrix}
\alpha_y & 0\\
0 & \alpha_z
\end{bmatrix}.
\end{equation*}
\end{lem}

\begin{proof}
In the main-lobe regime, the relevant pairwise displacements are
small enough for a second-order approximation of the steering
correlation.

For small $x$,
\begin{equation*}
\frac{1}{M^2}
\left|
\frac{\sin(Mx/2)}{\sin(x/2)}
\right|^2
\approx
1-\frac{M^2-1}{12}x^2.
\end{equation*}
Using $x_y=\pi\Delta y/D$ and $x_z=\pi\Delta z/D$, we obtain
\begin{align*}
\eta_y(\Delta y)
&\approx
1-
\frac{\pi^2(M_y^2-1)}{12D^2}
\Delta y^2
=
1-\alpha_y\Delta y^2,\\
\eta_z(\Delta z)
&\approx
1-
\frac{\pi^2(M_z^2-1)}{12D^2}
\Delta z^2
=
1-\alpha_z\Delta z^2.
\end{align*}
Therefore,
\begin{align*}
\eta(\Delta y,\Delta z)
&=
\eta_y(\Delta y)\eta_z(\Delta z)\\
&\approx
(1-\alpha_y\Delta y^2)(1-\alpha_z\Delta z^2)\\
&\approx
1-\alpha_y\Delta y^2-\alpha_z\Delta z^2,
\end{align*}
where the fourth-order term has been neglected.

Substituting this into \eqref{eq:B_delta_def}, we get
\begin{align*}
B(\Delta y,\Delta z)
&=
\log
\left[
1+
\frac{\gamma_0^2}{4(1+\gamma_0)}
\left(
1-\eta(\Delta y,\Delta z)
\right)
\right]\\
&\approx
\log
\left[
1+
\kappa
\left(
\alpha_y\Delta y^2+\alpha_z\Delta z^2
\right)
\right].
\end{align*}
In the main-lobe regime, the argument of the logarithm is close to one, so $\log(1+x)\approx x$, yielding \eqref{eq:B_quadratic}.
\end{proof}

Lemma~\ref{lem:quadratic_B} implies that the level sets of $B(\boldsymbol{\Delta})$ are approximately ellipses. 
Accordingly, the forbidden region in \eqref{eq:FJ_def} is approximated by
\begin{equation}\label{eq:ellipse_forbidden}
\mathcal{F}_J^{\rm ml}
=
\left\{
\boldsymbol{\Delta}\in\mathbb{R}^2:
\boldsymbol{\Delta}^{T}
\mathbf{G}_{\rm B}
\boldsymbol{\Delta}
<
B_J
\right\}.
\end{equation}
Its semi-axes are
\begin{equation*}
r_y(J)
=
\sqrt{\frac{B_J}{\kappa\alpha_y}},
\quad
r_z(J)
=
\sqrt{\frac{B_J}{\kappa\alpha_z}}.
\end{equation*}

The ellipse can be converted into a disk through the linear transformation
\begin{equation}\label{eq:linear_transform}
\widetilde{\mathbf{r}}
=
\mathbf{T}\mathbf{r},
~~
\mathbf{T}
=
\mathbf{G}_{\rm B}^{1/2}
=
\begin{bmatrix}
\sqrt{\kappa\alpha_y} & 0\\
0 & \sqrt{\kappa\alpha_z}
\end{bmatrix}.
\end{equation}

For a displacement vector 
$\boldsymbol{\Delta}=\mathbf{r}_i-\mathbf{r}_j$, 
its transformed displacement is $\widetilde{\boldsymbol{\Delta}}
\triangleq
\mathbf{T}\boldsymbol{\Delta}
=
\widetilde{\mathbf{r}}_i-\widetilde{\mathbf{r}}_j$.
Then $\boldsymbol{\Delta}^{T}\mathbf{G}_{\rm B}\boldsymbol{\Delta}
=
\|\widetilde{\boldsymbol{\Delta}}\|_2^2$,
and the forbidden ellipse becomes the disk
\begin{equation}\label{eq:transformed_forbidden_disk}
\widetilde{\mathcal{F}}_J^{\rm ml}
=
\left\{
\widetilde{\boldsymbol{\Delta}}:
\|\widetilde{\boldsymbol{\Delta}}\|_2^2<B_J
\right\}.
\end{equation}

\begin{thm}[Main-lobe hexagonal lattice design]\label{thm:hex_lattice_design}
Under the main-lobe approximation in Lemma~\ref{lem:quadratic_B}, the densest lattice codebook is obtained by placing a hexagonal lattice with minimum distance $\sqrt{B_J}$ in the transformed plane $\widetilde{\mathbf{r}}=\mathbf{T}\mathbf{r}$ and mapping it back to the original agent plane by $\mathbf{T}^{-1}$. 
Ignoring boundary effects, the resulting achievable alphabet size is approximately
\begin{equation}\label{eq:Jhex_Xi}
J_{\rm hex}
\approx
\frac{2\Xi L}{\sqrt{3} \log\frac{J-1}{\epsilon}},
\end{equation}
where
\begin{equation*}
\Xi
=
\frac{
\pi^2 a_y a_z \gamma_0^2
\sqrt{(M_y^2-1)(M_z^2-1)}
}{
48D^2(1+\gamma_0)
}.
\end{equation*}

For large $J$ and small $\epsilon$,
\begin{equation}\label{eq:Jhex_lb_closed}
J_{\rm hex}
\approx
\left\lfloor
\frac{
\Xi_{\rm h}L
}{
W_0\left(\frac{\Xi_{\rm h}L}{\epsilon}\right)
}
\right\rfloor,
\end{equation}
where $\Xi_{\rm h}=\frac{2\Xi}{\sqrt{3}}$ and $W_0$ is the principal branch of the Lambert-W function.
\end{thm}

\begin{proof}
In the transformed plane, the forbidden region is a disk of radius $\sqrt{B_J}$. 
Thus, a feasible lattice must have minimum Euclidean distance at least $\sqrt{B_J}$. 
Among two-dimensional lattices with a given minimum distance, the hexagonal lattice has the smallest fundamental cell area. 
For minimum distance $\sqrt{B_J}$, this cell area is $A_{\rm hex}(B_J)
=
\frac{\sqrt{3}}{2}B_J$.

The transformed agent plane has area
\begin{equation*}
|\widetilde{\mathcal{A}}|
=
|\det\mathbf{T}|\,|\mathcal{A}|
=
a_y a_z\sqrt{\det\mathbf{G}_{\rm B}}
\triangleq
\Xi.
\end{equation*}
Ignoring boundary effects, the number of retained lattice points is therefore
\begin{equation*}
J
\approx
\frac{
|\widetilde{\mathcal{A}}|
}{
A_{\rm hex}(B_J)
}
=
\frac{
2\Xi
}{
\sqrt{3}B_J
},
\end{equation*}
where $B_J
=
\frac{1}{L}
\log\frac{J-1}{\epsilon}$.

This expression is informative but implicit because $J$ appears on both sides. 
To obtain a conservative explicit expression, we use
\begin{equation*}
\log\frac{J-1}{\epsilon}
\approx
\log\frac{J}{\epsilon},
\end{equation*}
for large $J$ and small $\epsilon$. Then
\begin{equation*}
J
\approx
\frac{
2\Xi L
}{
\sqrt{3}\log\frac{J}{\epsilon}
}.
\end{equation*}
Solving the equality $J\log\frac{J}{\epsilon}=\frac{2}{\sqrt{3}}\Xi L$ gives a closed-form expression through the Lambert-$W$ function. 
Indeed, let $u=\log\frac{J}{\epsilon}$.
Then $J=\epsilon e^u$, and the equality becomes $\epsilon u e^u
\allowbreak =
\frac{2}{\sqrt{3}}\Xi L$.
Hence, we have
\begin{equation*}
u
=
W_0
\left(
\frac{2\Xi L}{\sqrt{3}\epsilon}
\right),
~~
J
=
\frac{
2\Xi L
}{\sqrt{3}
W_0\left(\frac{2\Xi L}{\sqrt{3}\epsilon}\right)
}.
\end{equation*}
Eq.~\eqref{eq:Jhex_lb_closed} follows by defining $\Xi_{\rm h}=\frac{2\Xi}{\sqrt{3}}$.
\end{proof}

Consequently, the achievable rate is
\begin{equation*}
R_{\rm hex}(L,\epsilon)
=
\frac{1}{LT_p}
\log_2
J_{\rm hex}.
\end{equation*}

Theorem~\ref{thm:hex_lattice_design} gives a closed-form geometric design rule under the main-lobe approximation. 
The BS sensing system induces a reliability metric $\mathbf{G}_{\rm B}$ on the agent plane. 
After whitening this metric, the approximately optimal lattice is hexagonal. 
Mapping it back to the original agent plane yields an inverse-scaled hexagonal lattice that is denser along directions where the BS has finer sensing resolution.

% The scalar $\Xi$ has a useful interpretation: it is the area of the agent plane measured in the BS-induced reliability metric. 
% It increases with the controllable physical area $a_y a_z$, the sensing SNR $\gamma_0$, and the array dimensions $M_y,M_z$, and decreases with the squared distance $D^2$.

\begin{rem}
It is worth noting that the result is approximate as it ignores sidelobes and boundary effects. 
For finite systems, the resulting lattice should be truncated to $\mathcal{A}$ and its exact reliability should be verified using the true field $B(\boldsymbol{\Delta})$. 
\end{rem}

\subsection{Snapshot-Number Optimization}\label{subsec:L_optimization}
In the considered multi-snapshot model, the number of sensing snapshots $L$ controls the sensing resources allocated to each embodied symbol. Increasing $L$ improves reliability because pairwise errors decay as $\exp[-LB(\Delta)]$. At the same time, one embodied channel use occupies $LT_p$ seconds, reducing the rate normalization factor. Thus, a tradeoff arises between sensing reliability and observation duration.

The explicit lower bound in \eqref{eq:Jhex_lb_closed} makes this tradeoff transparent. 
Ignoring the floor operation, the rate is
\begin{equation}\label{eq:R_L_approx_hex}
R_{\rm hex}(L,\epsilon)
\approx
\frac{1}{LT_p}
\log_2
\left[
\frac{
\Xi_{\rm h}L
}{
W_0\left(\frac{\Xi_{\rm h}L}{\epsilon}\right)
}
\right].
\end{equation}
The term inside the logarithm increases with $L$, because more snapshots allow more reliable position discrimination and hence a denser physical alphabet. 
However, the prefactor $1/(LT_p)$ decreases with $L$. 
The balance between these two effects yields a finite optimal sensing duration.

\begin{thm}[Optimal snapshot number]\label{thm:optimal_L}
Under the main-lobe hexagonal lattice approximation, the continuous relaxation of $\max_{L>0}
R_{\rm hex}(L,\epsilon)$ has the stationary point
\begin{equation}\label{eq:L_star_approx_hex}
L_{\rm cont}^{\star}
\approx
\frac{\epsilon}{\Xi_{\rm h}}
y^\star e^{y^\star},
\end{equation}
where $y^\star
=
\frac{
q+\sqrt{q^2+4q}
}{2}$ and $q
=-\log{\epsilon}$.
For integer snapshot numbers, one may evaluate $L\!=\!\lfloor L_{\rm cont}^{\star}\rfloor$ and $L\!=\!\lceil L_{\rm cont}^{\star}\rceil$ using the exact reliability condition based on the true field $B(\boldsymbol{\Delta})$.
\end{thm}

\begin{proof}
Let $x
=
\frac{\Xi_{\rm h}L}{\epsilon}$ and $y
=
W_0(x)$.
Then $x=ye^y$, and hence $L
=
\frac{\epsilon}{\Xi_{\rm h}}ye^y$. Therefore,
\begin{equation*}
J_{\rm hex}
\approx
\frac{\Xi_{\rm h}L}{W_0(\Xi_{\rm h}L/\epsilon)}
=
\epsilon e^y.
\end{equation*}
Up to the constant factor $1/(T_p\log 2)$, maximizing \eqref{eq:R_L_approx_hex} is equivalent to maximizing
\begin{equation*}
g(y)
=
\frac{
\log J_{\rm hex}
}{
L
}
=
\frac{
\log\epsilon+y
}{
(\epsilon/\Xi_{\rm h})ye^y
}.
\end{equation*}
The constant $\epsilon/\Xi_{\rm h}$ does not affect the optimizer. 
Let $q=-\log\epsilon$. 
Then the objective is proportional to
\begin{equation*}
\widetilde{g}(y)
=
\frac{y-q}{ye^y}.
\end{equation*}
The feasible rate requires $y>q$. 
Taking the derivative gives
\begin{equation*}
\frac{d}{dy}
\left[
\frac{y-q}{ye^y}
\right]
=
e^{-y}
\left[
\frac{q}{y^2}
-
\frac{y-q}{y}
\right].
\end{equation*}
Setting the derivative to zero yields $\frac{q}{y^2}
=
\frac{y-q}{y}$.
The positive root is
\begin{equation*}
y^\star
=
\frac{
q+\sqrt{q^2+4q}
}{2}.
\end{equation*}
Substituting $L=(\epsilon/\Xi_{\rm h})ye^y$ yields \eqref{eq:L_star_approx_hex}.
\end{proof}

Theorem~\ref{thm:optimal_L} shows that increasing the sensing duration is not always beneficial.
The result also reveals the role of system parameters. 
A larger transformed-area parameter $\Xi_{\rm h}$ reduces the required optimal snapshot number because the BS-agent geometry already supports a dense physical alphabet. 
In contrast, a smaller $\Xi_{\rm h}$, caused for example by larger distance $D$, smaller sensing SNR $\gamma_0$, or smaller array aperture, requires more snapshots to achieve reliable discrimination.
\section{Upper Bounds on $\epsilon$-Capacity}\label{sec:upper_bounds}

Section~\ref{sec:achievable} constructs reliable embodied codebooks and thereby establishes achievable lower bounds on $C_{\epsilon}(L)$. 
We now derive upper bounds that limit how large the embodied alphabet can be.
We present two complementary converses. 
The first is an information-theoretic converse based on Fano's inequality. 
It upper-bounds the number of reliably distinguishable physical positions by the amount of statistical information that the BS sensing observations can carry about the message. 
The second is a geometric packing converse. 
It translates binary testing impossibility into a necessary physical separation between codewords, and then upper-bounds how many such separated positions can fit inside the agent plane.

\subsection{Information-Theoretic Converse}\label{subsec:info_converse}

$C_{\epsilon}(L)$ is defined through the largest number of messages that can be decoded with maximum error probability no larger than $\epsilon$. 
Fano's inequality applies to any such reliable decoding problem with a discrete message, an observation, and a decoder \cite{cover1999elements}. 
Therefore, it can be used to upper-bound the maximum feasible embodied alphabet size.

We first define the sensing covariance set induced by the agent plane:
\begin{equation*}
\mathcal{Q}_{\mathcal{A}}
=
\operatorname{conv}
\left\{
\mathbf{a}(\mathbf{r})\mathbf{a}^{H}(\mathbf{r})
:
\mathbf{r}\in\mathcal{A}
\right\}.
\end{equation*}
Each element of $\mathcal{Q}_{\mathcal{A}}$ is a convex mixture of rank-one spatial covariance matrices generated by feasible physical positions. 
This set captures the fact that the embodied alphabet cannot excite arbitrary array covariance directions; it is constrained by the finite physical support $\mathcal{A}$.

Define the support-constrained single-snapshot information upper bound
\begin{equation}\label{eq:Csnap_A_ub_def}
C_{\rm snap}^{\rm ub}(\mathcal{A})
=
\sup_{\mathbf{Q}\in\mathcal{Q}_{\mathcal{A}}}
\log_2
\frac{
\det\left(
\mathbf{I}_M+\gamma_0\mathbf{Q}
\right)
}{
1+\gamma_0
}.
\end{equation}

\begin{thm}[Support-constrained information converse]\label{thm:support_fano_upper}
The $\epsilon$-capacity of the embodied channel satisfies
\begin{equation}\label{eq:support_info_upper_capacity}
C_{\epsilon}(L)
\leq
C_{\rm info}^{\rm ub}(L,\epsilon)
\triangleq
\frac{
C_{\rm snap}^{\rm ub}(\mathcal{A})
+
\frac{1}{L}h_2(\epsilon)
}{
(1-\epsilon)T_p
},
\end{equation}
where $h_2(\epsilon)
=
-\epsilon\log_2\epsilon
-
(1-\epsilon)\log_2(1-\epsilon)$ is the binary entropy function.
\end{thm}

\begin{proof}[Proof sketch]
Consider any $\epsilon$-reliable codebook of size $J$ and take $W$ to be uniform over the codebook. 
Since maximum error reliability implies average error reliability, Fano's inequality gives $(1-\epsilon)\log_2 J
\leq
I(W;\mathbf{Y})+h_2(\epsilon)$.

For any feasible codebook, the induced average spatial covariance can be written as $\mathbf{Q}_{\mathfrak{C}}
=
\frac{1}{J}
\sum_{j=1}^{J}
\mathbf{a}(\mathbf{r}_j)\mathbf{a}^H(\mathbf{r}_j)
\in
\mathcal{Q}_{\mathcal{A}}$.
Using the maximum-entropy property of the complex Gaussian distribution, we can upper bound the mutual information by $I(W;\mathbf{Y})
\leq L C_{\rm snap}^{\rm ub}(\mathcal{A})$.
Substituting this into Fano's inequality yields the converse.
\end{proof}

\subsection{Geometric Packing Converse}\label{subsec:geo_converse}

The information-theoretic converse limits the amount of information that can be extracted from the sensing observations. 
We next derive a complementary converse based on the physical geometry of the agent plane. 
% The intuition is: if two physical positions are too close in the sensing geometry, then no binary detector can distinguish them with sufficiently small error. 
% Thus, any maximum-error reliable codebook must have a minimum physical separation, which limits the number of codewords that can be placed in $\mathcal{A}$.
We first establish a necessary pairwise condition.

\begin{lem}[Necessary pairwise Bhattacharyya separation]\label{lem:necessary_pairwise}
Consider two codeword positions $\mathbf{r}_i$ and $\mathbf{r}_j$ with displacement $\boldsymbol{\Delta}_{ij}$.
Suppose that the BS can distinguish these two positions from $L$ snapshots with error probability at most $\epsilon$ under both hypotheses. 
Their single-snapshot Bhattacharyya distance must satisfy
\begin{equation}\label{eq:Bnec_condition}
B(\boldsymbol{\Delta}_{ij})
\geq
B_{\rm nec}(\epsilon,L),
\end{equation}
where $B_{\rm nec}(\epsilon,L)
=
\frac{1}{2L}
\log
\frac{1}{4\epsilon(1-\epsilon)}$.
\end{lem}

\begin{proof}
Let $P_i^{(L)}$ and $P_j^{(L)}$ denote the $L$-snapshot distributions induced by $\mathbf{r}_i$ and $\mathbf{r}_j$, respectively. 
Their Bhattacharyya coefficient is $\beta_{ij}^{(L)}
=
\exp[-L B(\boldsymbol{\Delta}_{ij})]$.
For equal priors, the minimum binary Bayes error probability satisfies
\begin{equation*}
P_{\rm b}^{\star}
=
\frac{1}{2}
\left(
1-
\|P_i^{(L)}-P_j^{(L)}\|_{\rm TV}
\right),
\end{equation*}
where $\|\cdot\|_{\rm TV}$ denotes total variation distance. 
Using the standard relation between total variation distance and Bhattacharyya coefficient,
\begin{equation*}
\|P_i^{(L)}-P_j^{(L)}\|_{\rm TV}
\leq
\sqrt{
1-
\left(\beta_{ij}^{(L)}\right)^2
},
\end{equation*}
we obtain
\begin{equation*}
P_{\rm b}^{\star}
\geq
\frac{1}{2}
\left(
1-
\sqrt{
1-\exp[-2L B(\boldsymbol{\Delta}_{ij})]
}
\right).
\end{equation*}
If a binary test has error probability no larger than $\epsilon$ under both hypotheses, then its equal-prior average error probability is also no larger than $\epsilon$. 
Therefore, it is necessary that
\begin{equation*}
\frac{1}{2}
\left(
1-
\sqrt{
1-\exp[-2L B(\boldsymbol{\Delta}_{ij})]
}
\right)
\leq
\epsilon.
\end{equation*}
Rearranging gives $\exp[-2L B(\boldsymbol{\Delta}_{ij})]
\leq
4\epsilon(1-\epsilon)$,
which is equivalent to \eqref{eq:Bnec_condition}.
\end{proof}

Lemma~\ref{lem:necessary_pairwise} states that reliable pairwise discrimination requires a minimum amount of sensing distinguishability. 
Accordingly, define the necessary forbidden region
\begin{equation}\label{eq:Fnec_def}
\mathcal{F}_{\rm nec}(\epsilon,L)
=
\left\{
\boldsymbol{\Delta}\in\mathbb{R}^2:
B(\boldsymbol{\Delta})
<
B_{\rm nec}(\epsilon,L)
\right\}.
\end{equation}
Any pair of codewords in an $\epsilon$-reliable maximum-error codebook must have a displacement outside $\mathcal{F}_{\rm nec}(\epsilon,L)$.

The exact shape of $\mathcal{F}_{\rm nec}(\epsilon,L)$ is generally anisotropic and may be nonconvex because of array sidelobes. 
For an explicit geometric upper bound, we extract from it a Euclidean necessary separation. 
Define
\begin{equation*}
d_{\rm nec}(\epsilon,L)
=
\sup
\left\{
d\geq 0:
\{\boldsymbol{\Delta}:\|\boldsymbol{\Delta}\|_2\leq d\}
\subseteq
\mathcal{F}_{\rm nec}(\epsilon,L)
\right\}.
\end{equation*}
Thus, $d_{\rm nec}(\epsilon,L)$ is the radius of the largest Euclidean ball around the origin that is entirely contained in the necessary forbidden region. 

\begin{thm}[Geometric packing upper bound]\label{thm:geometric_upper}
Assume $\epsilon<1/2$. 
The embodied $\epsilon$-capacity satisfies
\begin{eqnarray*}
&& C_{\epsilon}(L)
\leq
\frac{1}{LT_p}
\log_2
\Bigg[
\frac{4}{\pi d_{\rm nec}^2(\epsilon,L)}
\Big(
a_y a_z \\
&&\hspace{1.5cm} 
+
(a_y+a_z)d_{\rm nec}(\epsilon,L)
+
\frac{\pi d_{\rm nec}^2(\epsilon,L)}{4}
\Big)
\Bigg].
\end{eqnarray*}
\end{thm}

\begin{proof}
Consider any $\epsilon$-reliable embodied codebook $\mathfrak{C}$ with a decoder $g$ such that $\Pr(g(\mathbf{Y})\neq i|W=i)
\leq
\epsilon$, $\forall i$.
The multi-hypothesis decoder $g$ can be converted into a binary test between hypotheses $i$ and $j$ by declaring hypothesis $i$ if $g(\mathbf{Y})=i$, and declaring hypothesis $j$ otherwise. 
Under hypothesis $i$, this binary test makes an error only when the original decoder fails to output $i$, so its error probability is at most $\epsilon$. 
Under hypothesis $j$, the binary test declares $i$ only when the original decoder outputs $i$, which is contained in the event that the original decoder fails to output $j$. 
Thus, its error probability under hypothesis $j$ is also at most $\epsilon$.

By Lemma~\ref{lem:necessary_pairwise}, every pair of distinct codewords must satisfy
\begin{equation*}
B(\mathbf{r}_i-\mathbf{r}_j)
\geq
B_{\rm nec}(\epsilon,L).
\end{equation*}
That is, $\mathbf{r}_i-\mathbf{r}_j
\notin
\mathcal{F}_{\rm nec}(\epsilon,L)$, $\forall i\neq j$.

By the definition of $d_{\rm nec}(\epsilon,L)$, any displacement with Euclidean norm no larger than $d_{\rm nec}(\epsilon,L)$ lies inside $\mathcal{F}_{\rm nec}(\epsilon,L)$. 
Therefore,
\begin{equation*}
\|\mathbf{r}_i-\mathbf{r}_j\|_2
>
d_{\rm nec}(\epsilon,L),
\qquad
\forall i\neq j.
\end{equation*}

Now place a disk of radius $d_{\rm nec}(\epsilon,L)/2$ centered at each codeword. 
These disks are non-overlapping. 
Since the codeword centers lie in the rectangle $\mathcal{A}$, all disks are contained in the expanded region
\begin{equation*}
\mathcal{A}\oplus\mathcal{B}
\left(
0,\frac{d_{\rm nec}(\epsilon,L)}{2}
\right),
\end{equation*}
where $\oplus$ denotes the Minkowski sum \cite{varadhan2004accurate}. 
The area of this expanded region is
\begin{equation*}
a_y a_z
+
(a_y+a_z)d_{\rm nec}(\epsilon,L)
+
\frac{\pi d_{\rm nec}^2(\epsilon,L)}{4}.
\end{equation*}
Each disk has area $\frac{\pi d_{\rm nec}^2(\epsilon,L)}{4}$.
Because the disks are non-overlapping, their total area cannot exceed the area of the expanded region. 
Therefore,
\begin{equation*}
J
\leq
\frac{
a_y a_z
+
(a_y+a_z)d_{\rm nec}(\epsilon,L)
+
\frac{\pi d_{\rm nec}^2(\epsilon,L)}{4}
}{
\frac{\pi d_{\rm nec}^2(\epsilon,L)}{4}
}.
\end{equation*}
Taking $\log_2(\cdot)/(LT_p)$ and maximizing over all $\epsilon$-reliable codebooks gives the upper bound.
\end{proof}

The bound in Theorem~\ref{thm:geometric_upper} is explicit once $d_{\rm nec}(\epsilon,L)$ is known. 
It is conservative because it uses the largest Euclidean ball contained in the necessary forbidden region rather than the full anisotropic shape of $\mathcal{F}_{\rm nec}(\epsilon,L)$. 
We now give a closed-form approximation for $d_{\rm nec}(\epsilon,L)$ under the main-lobe model of Lemma~\ref{lem:quadratic_B}.

\begin{cor}[Main-lobe geometric upper bound]\label{cor:ml_geo_upper}
Let $\alpha_{\max}=\max\{\alpha_y,\alpha_z\}$.
Under the main-lobe quadratic approximation, the necessary Euclidean separation is approximated by
\begin{equation}\label{eq:dnec_ml}
d_{\rm nec}^{\rm ml}(\epsilon,L)
\approx
\sqrt{
\frac{1}{
2\kappa L\alpha_{\max}
}
\log
\frac{1}{4\epsilon(1-\epsilon)}
}.
\end{equation}

Substituting $d_{\rm nec}^{\rm ml}(\epsilon,L)$ into the geometric packing upper bound gives the closed-form approximation
\begin{eqnarray*}\label{eq:geo_upper_ml_closed}
&& C_{\epsilon}(L)
\lesssim
\frac{1}{LT_p}
\log_2
\Bigg[
1
+
\frac{4(a_y+a_z)}{\pi}
\sqrt{
\frac{
\Omega_{\max}L
}{
\log\frac{1}{4\epsilon(1-\epsilon)}
}
}
+ \\
&&\hspace{2cm}
\frac{4a_y a_z}{\pi}
\frac{
\Omega_{\max}L
}{
\log\frac{1}{4\epsilon(1-\epsilon)}
}
\Bigg],
\end{eqnarray*}
where $M_{\max}
\!\triangleq\!
\max\{M_y^2\!-\! 1,M_z^2\!-\! 1\}$
and $\Omega_{\max}
\triangleq
\frac{
\pi^2\gamma_0^2 M_{\max}
}{
24D^2(1+\gamma_0)
}$.
\end{cor}

% \begin{proof}
% Under the main-lobe approximation,
% \begin{equation*}
% B(\Delta y,\Delta z)
% \approx
% \kappa
% \left(
% \alpha_y\Delta y^2+\alpha_z\Delta z^2
% \right).
% \end{equation*}
% For a fixed Euclidean radius $d$, the maximum value of this quadratic form over all directions is
% \begin{equation*}
% \max_{\|\boldsymbol{\Delta}\|_2=d}
% B(\boldsymbol{\Delta})
% \approx
% \kappa\alpha_{\max}d^2.
% \end{equation*}
% If $\kappa\alpha_{\max}d^2 < B_{\rm nec}(\epsilon,L)$,
% then even the best-resolved direction at radius $d$ cannot meet the necessary binary distinguishability threshold. 
% Hence all displacements of norm no larger than $\sqrt{
% \frac{
% B_{\rm nec}(\epsilon,L)
% }{
% \kappa\alpha_{\max}
% }
% }$ 
% belong to the necessary forbidden region under the quadratic approximation. 
% This gives \eqref{eq:dnec_ml}. 
% \end{proof}

For achievability, a conservative Euclidean design must protect against the worst-resolved direction of the BS sensing field. 
For the converse, however, a pair of positions can be ruled out only if even the best-resolved direction cannot provide enough distinguishability. 
This is why the closed-form necessary separation in \eqref{eq:dnec_ml} is governed by $\alpha_{\max}$. 
\section{Numerical Results}\label{sec:numerical}

\begin{figure}[t]
    \centering
    \includegraphics[width=1\columnwidth]{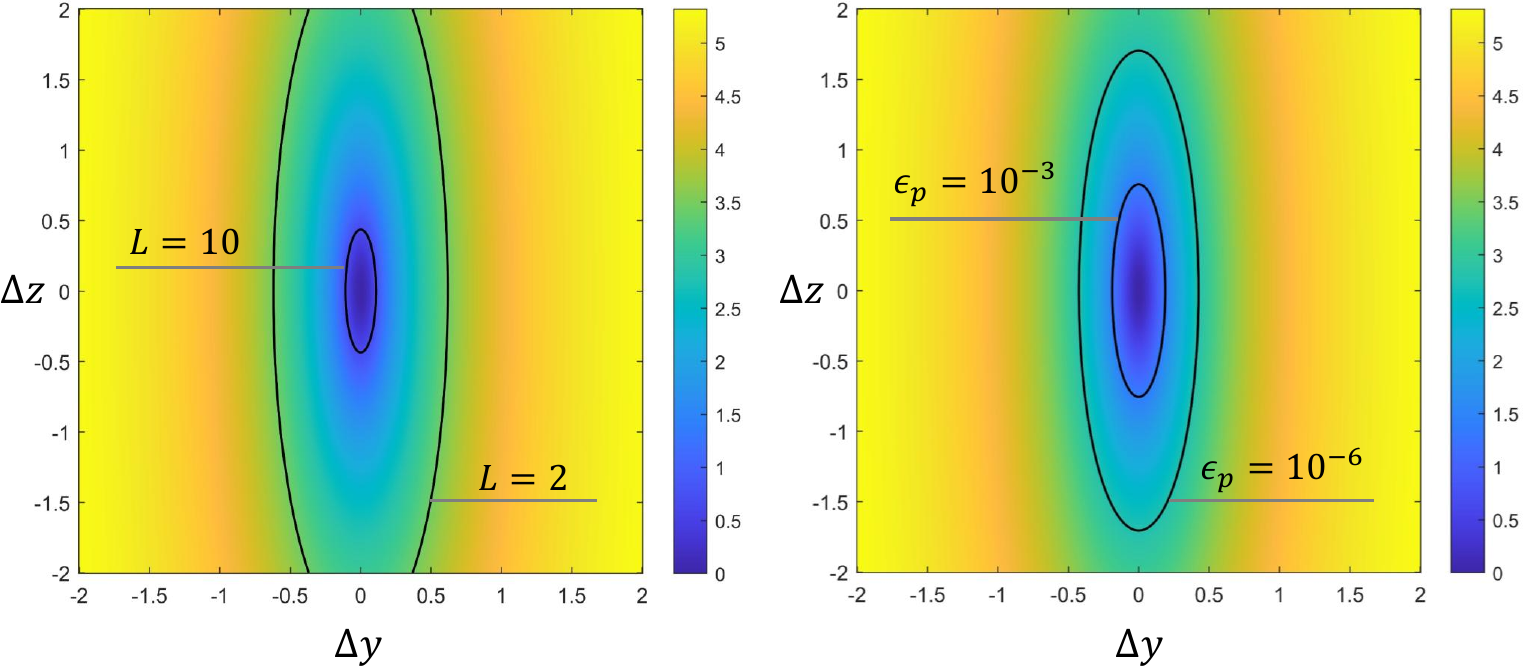}
    \caption{Sensing-induced reliability field $B(\boldsymbol{\Delta})$ determines the pairwise distinguishability of embodied symbols.}
    \label{fig:reliability_geometry}
\end{figure}

This section numerically evaluates the sensing-induced geometry and the achievable performance of embodied communication. 
Unless otherwise specified, the carrier frequency is $f_c=7$ GHz, the BS is equipped with a $64\times 16$ UPA, the BS-agent distance is $D=100$ m, the number of snapshots $L=5$, and the target maximum error probability is $\epsilon=10^{-3}$. 
The controllable agent region is a square with $a_y=a_z=2$ m. 
To remove the dependence on the particular probing-pulse duration, we report the rate in normalized units of bits per probing duration, i.e., $R=\frac{1}{L}\log_2 J$.

\subsection{Sensing-Induced Reliability Fields}

We first examine the reliability field $B(\boldsymbol{\Delta})$ in \eqref{eq:B_delta_def} induced by the BS sensing system on the agent plane. 
This field maps a physical displacement between two possible scatterer positions to the single-snapshot Bhattacharyya distance of the corresponding sensing distributions. 
It is therefore the basic geometric object that determines how physical positions can be used as embodied symbols.

Fig.~\ref{fig:reliability_geometry} shows the two-dimensional field $B(\boldsymbol{\Delta})$ together with the contours corresponding to different snapshot numbers $L$ and target pairwise error probabilities $\epsilon_{\rm p}$.
\begin{itemize}[leftmargin=0.5cm]
    \item For a fixed $\epsilon_{\rm p}$, increasing $L$ reduces the required single-snapshot exponent $B_{\rm req}=\frac{1}{L}\log\frac{1}{\epsilon_{\rm p}}$, and therefore shrinks the forbidden displacement region. More snapshots allow the BS to distinguish closer physical positions and hence support denser embodied codebooks.
    \item For a fixed $L$, decreasing $\epsilon_{\rm p}$ has the opposite effect: it requires a larger Bhattacharyya exponent, enlarges the forbidden region, and reduces the number of admissible codewords. This directly illustrates the operational meaning of the $\epsilon$-capacity in Section~\ref{sec:capacity}.
\end{itemize}

Fig.~\ref{fig:reliability_geometry2} further plots the pairwise error bound $\exp[-LB(\boldsymbol{\Delta})]$ as a function of Euclidean displacement along representative directions. 
The decay rates are direction-dependent, which means that two codeword pairs with the same Euclidean separation can have different sensing distinguishability.
This behavior is caused by the anisotropic UPA aperture and confirms that the relevant codebook geometry is not ordinary Euclidean distance, but the sensing-induced field $B(\boldsymbol{\Delta})$.

The directional polar plot in Fig.~\ref{fig:reliability_geometry2} provides another view of the same phenomenon by plotting $B(d\cos\psi,d\sin\psi)$ against the direction angle $\psi$ for a fixed displacement radius $d$. 
The anisotropic polar profile reflects the unequal resolution of the UPA along different directions. 
This is precisely why the lattice codebook in Section~\ref{sec:achievable} is designed in the transformed reliability domain rather than directly as an isotropic grid on the physical plane.

\begin{figure}[t]
    \centering
    \includegraphics[width=1\columnwidth]{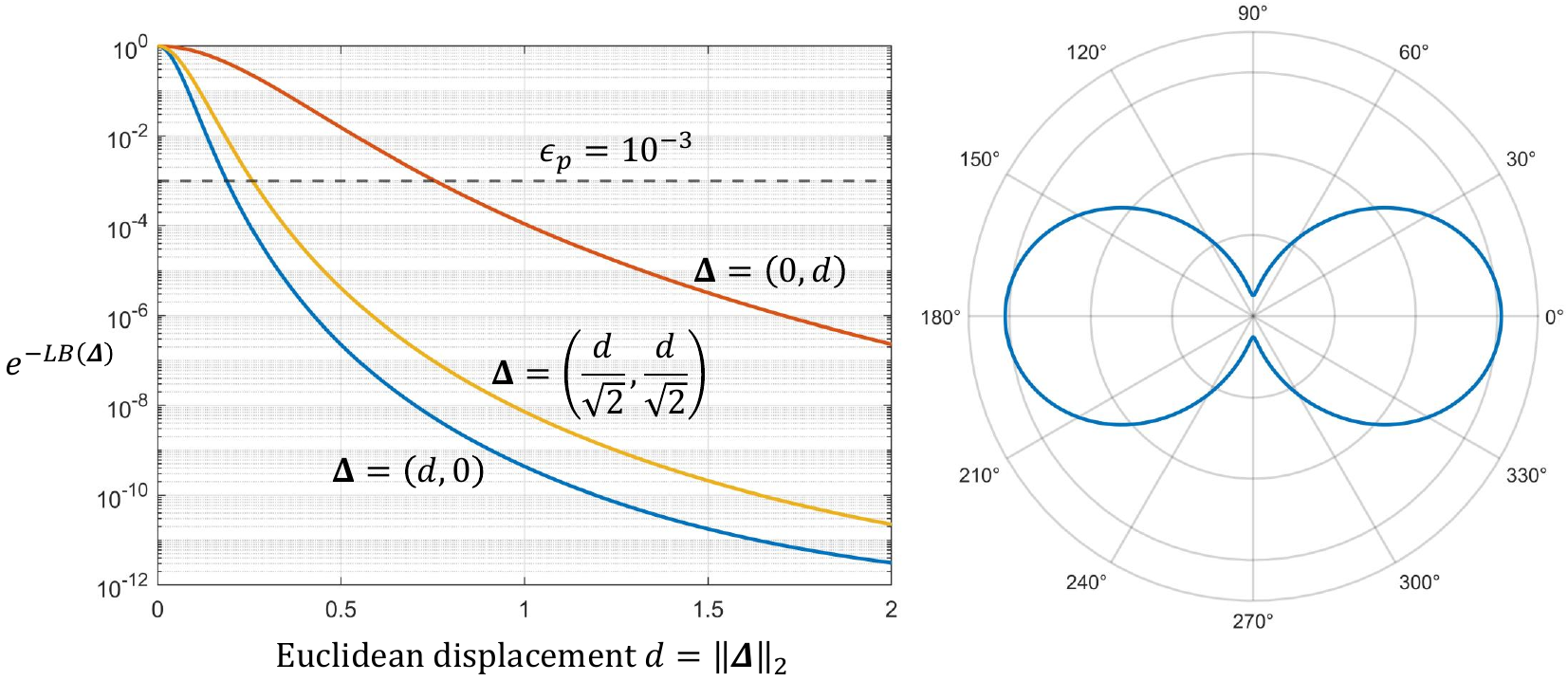}
    \caption{Directional pairwise error behavior. 
    The same Euclidean displacement may lead to different error exponents.}
    \label{fig:reliability_geometry2}
\end{figure}

\begin{figure}[t]
    \centering
    \includegraphics[width=0.6\columnwidth]{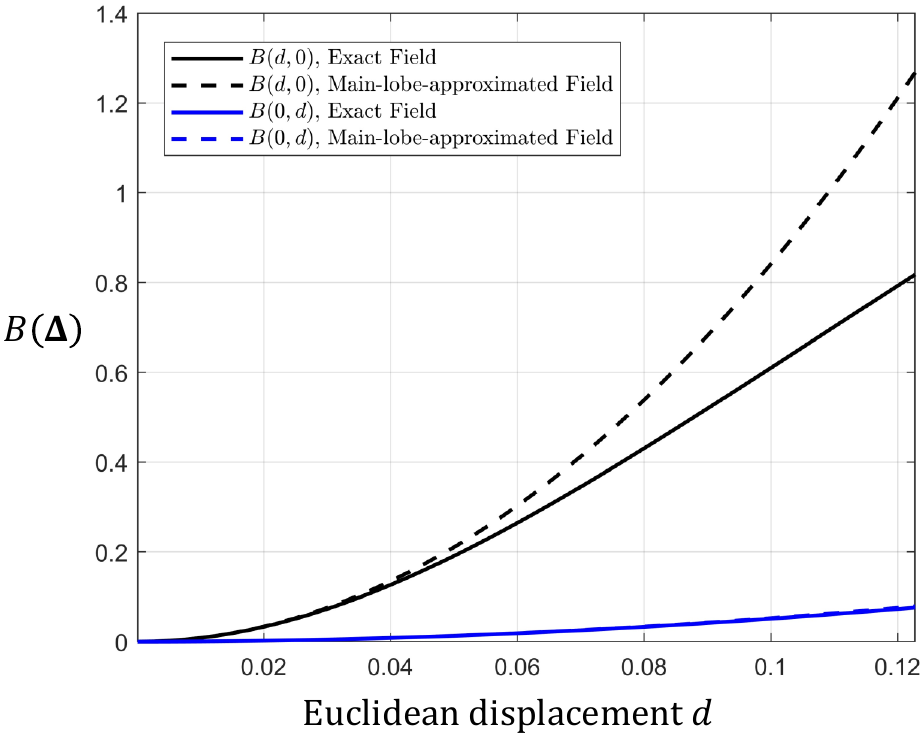}
    \caption{Exact Bhattacharyya distance versus the main-lobe quadratic approximation.}
    \label{fig:main_lobe_accuracy}
\end{figure}

We next validate the main-lobe approximation used to derive the closed-form lattice design in Section~\ref{sec:achievable}. 
Recall that the exact reliability field is induced by the UPA Dirichlet kernels, while the main-lobe approximation replaces it locally by the quadratic form $B(\boldsymbol{\Delta}) \approx \boldsymbol{\Delta}^{T}\mathbf{G}_{\rm B}\boldsymbol{\Delta}$.
This approximation underlies the transformed hexagonal lattice construction in Theorem~\ref{thm:hex_lattice_design}.

Fig.~\ref{fig:main_lobe_accuracy} compares the exact Bhattacharyya distance with its quadratic approximation along representative displacement directions. 
The approximation is accurate near the origin, which is the dominant regime for dense codebook packing. 
For larger displacements, the exact field deviates from the quadratic model because of sidelobes and nonlocal array-response effects.

\subsection{Achievable Rate and Converses}

We now evaluate the achievable rate of the proposed hexagonal lattice codebook and compare it with the converses established in Section~\ref{sec:upper_bounds}. 

\begin{figure}[t]
    \centering
    \includegraphics[width=1\columnwidth]{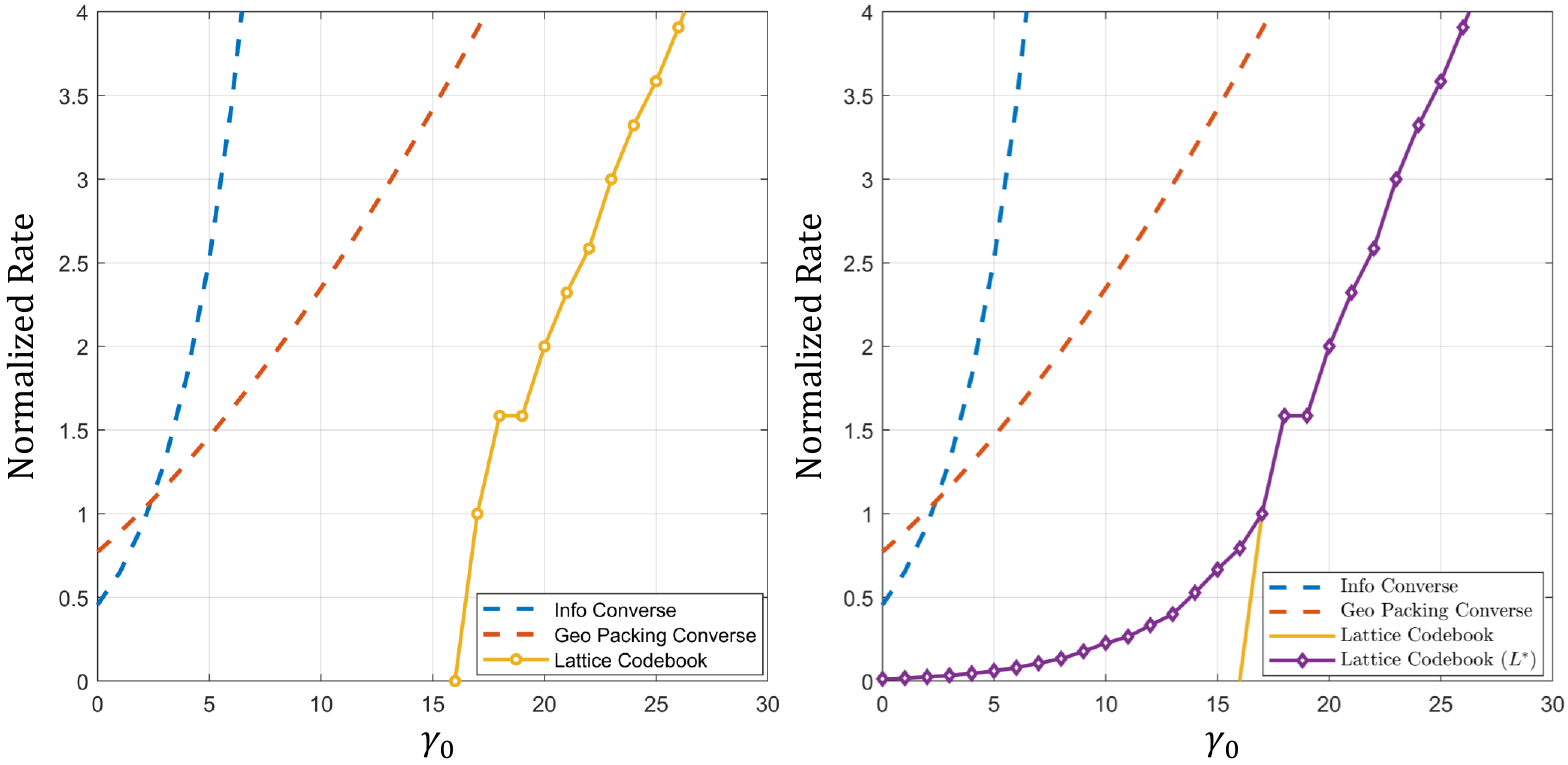}
    \caption{Achievable normalized rate versus sensing SNR, where $M_y\!=\! M_z\!=\! 64$ and $L\!=\! 1$. Optimizing over $L$ improves the rate by balancing sensing reliability and symbol duration.}
    \label{fig:rate_snr}
\end{figure}

Fig.~\ref{fig:rate_snr} contains two subplots. 
We first focus on the left subplot, where the number of snapshots is fixed as $L=1$. 
As the sensing SNR increases, the achievable rate increases. 
This is consistent with the analytical dependence of the reliability field on $\gamma_0$: a larger SNR increases the Bhattacharyya distance between nearby positions, reduces the effective forbidden region, and allows more physical positions to be reliably used as codewords.

However, the fixed-$L$ result only captures one operating point. 
As discussed in Section~\ref{subsec:L_optimization}, the number of snapshots itself is a key design variable in embodied communication. 
Increasing $L$ improves sensing reliability because pairwise error probabilities decay exponentially with $L$. 
At the same time, one embodied symbol occupies $LT_p$ seconds, so the normalized rate is penalized by the factor $1/L$. 
This creates a sensing-duration tradeoff: a larger $L$ can support a larger alphabet, but too large an $L$ reduces the rate per unit probing duration.

To make this tradeoff explicit, Fig.~\ref{fig:Lstar_snr} plots the normalized rate as a function of $L$ for several SNR values. 
As shown, 
\begin{itemize}[leftmargin=0.5cm]
    \item The variation of $L$ creates a multimodal rate curve and there exists an optimal sensing duration.
    \item Higher SNR shifts the optimum toward smaller $L$. This is because each snapshot is more informative at high SNR, hence fewer snapshots are sufficient to resolve a dense physical alphabet. At lower SNR, more snapshots are needed to accumulate enough sensing evidence.
    \item Since the rate is normalized by $1/L$, excessive sensing duration eventually reduces the rate. 
\end{itemize}

Fig.~\ref{fig:Lstar_snr} further plots the optimized snapshot number $L^\star$ versus SNR. 
As the SNR increases, fewer snapshots are needed to achieve the same level of position distinguishability. 
This agrees with the closed-form expression in Theorem~\ref{thm:optimal_L}, where a larger transformed-area parameter $\Xi_{\rm h}$ reduces the optimal sensing duration. 
We point out that small fluctuations of the optimized integer-valued $L^\star$ can occur because both $L$ and the supported alphabet size are discrete, and because finite-region truncation affects the realized lattice codebook.

With the role of $L$ clarified, we return to the right subplot of Fig.~\ref{fig:rate_snr}, where $L$ is optimized for each SNR. 
Compared with the fixed-$L$ case, optimizing the sensing duration improves the achievable rate, especially in the low-SNR regimes where a single snapshot is insufficient to resolve a dense alphabet.

Finally, Fig.~\ref{fig:codebook_layout} visualizes a generated hexagonal lattice codebook on the agent plane.
Because the BS-induced reliability metric is anisotropic, the resulting codebook is not an ordinary isotropic grid in the physical domain. 
Instead, it is stretched and compressed according to the sensing resolution of the UPA. 
In directions where the BS provides finer discrimination, the codewords can be placed more densely; in weaker directions, the spacing is larger.

This layout gives a direct physical interpretation of embodied communication. 
The agent plane becomes a physical alphabet, and the BS sensing system determines how finely that alphabet can be carved.

\begin{figure}[t]
    \centering
    \includegraphics[width=1\columnwidth]{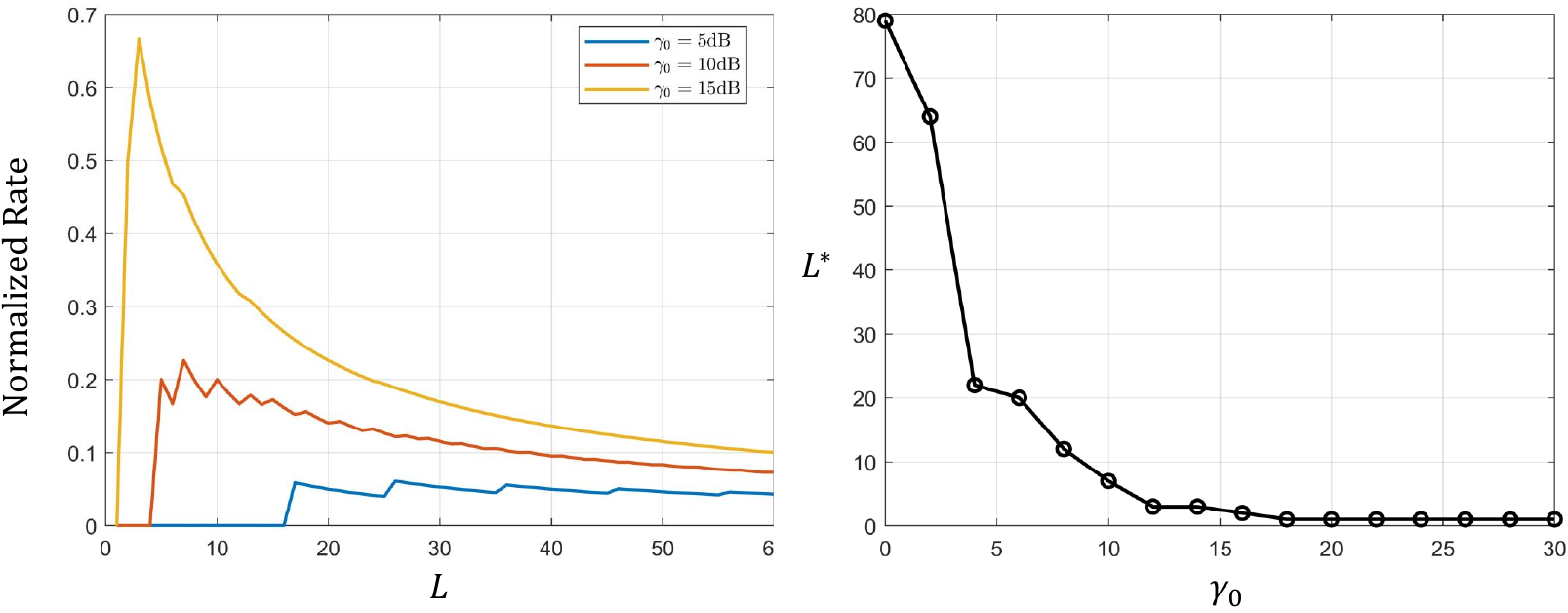}
    \caption{Normalized rate as a function of $L$, and the optimal snapshot number versus sensing SNR.}
    \label{fig:Lstar_snr}
\end{figure}

\begin{figure}[t]
    \centering
    \includegraphics[width=1\columnwidth]{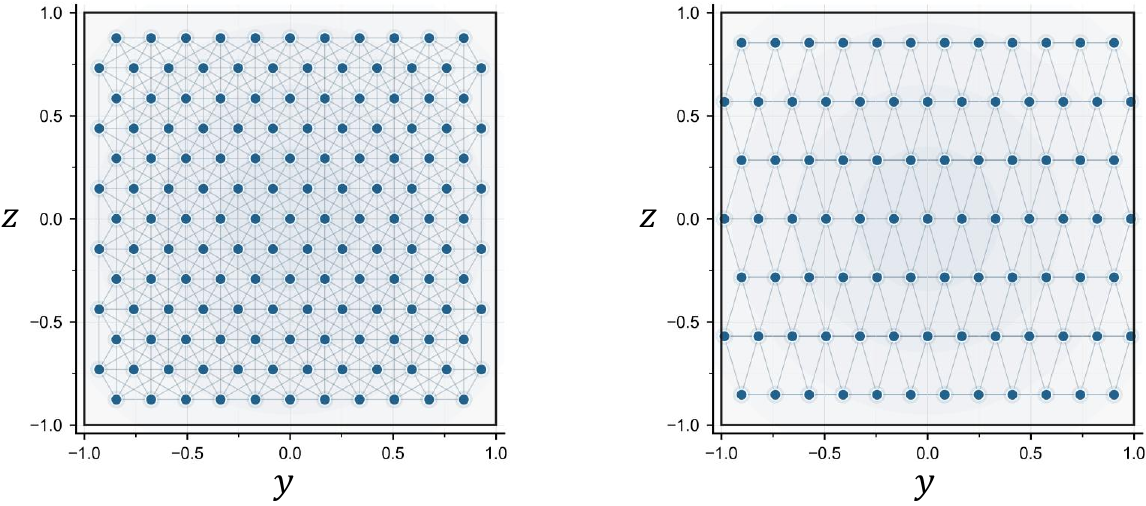}
    \caption{Example hexagonal lattice codebooks. Left: $M_y=M_z=64$, right: $M_y=64$ and $M_z=32$.}
    \label{fig:codebook_layout}
\end{figure}

\subsection{Beyond the Main-Lobe Approximation}
\label{subsec:beyond_main_lobe}

\begin{figure}[t]
    \centering
    \includegraphics[width=0.7\columnwidth]{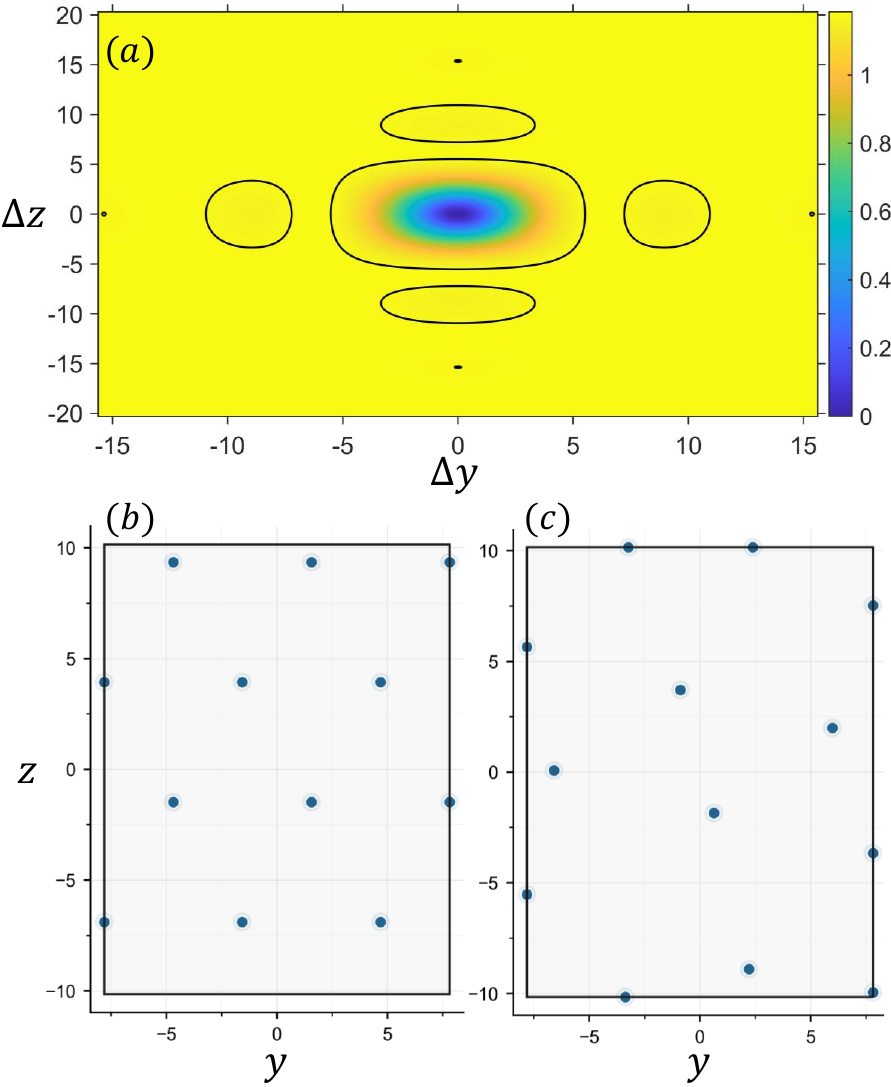}
    \caption{
    {
    Codebook design beyond the main-lobe approximation.
    (a) $B(\boldsymbol{\Delta})$ over the enlarged displacement region.
    (b) A hexagonal codebook constructed using the main-lobe approximation.
    (c) A nonuniform codebook obtained by solving the forbidden-region packing problem using the complete field.
    }
    }
    \label{fig:general}
\end{figure}

The preceding numerical results have illustrated how sensing resolution shapes the embodied alphabet. In the operating regime considered there, the main-lobe approximation provides a useful description of the local reliability geometry and leads to the regular hexagonal construction illustrated in Fig.~\ref{fig:codebook_layout}. The situation can be different when an agent can realize physical states over a wider region: codeword separations may then extend beyond the main lobe, where the quadratic approximation no longer captures the complete sensing response. We now examine such a regime to understand how the reliability geometry and the resulting codebook design change.

We consider a $64\times64$ UPA with BS-agent distance $D=200\,\mathrm{m}$, sensing SNR of $10\,\mathrm{dB}$, $L=8$, and target $\epsilon=10^{-3}$. The controllable region is $16\times20\,\mathrm{m}^2$, spanning well beyond the main lobe.
Fig.~\ref{fig:general}(a) shows the reliability field $B(\boldsymbol{\Delta})$ over the enlarged displacement region. While the contours in Fig.~\ref{fig:reliability_geometry} appear approximately elliptical within the local region considered there, the complete field now reveals additional disconnected forbidden regions at larger displacements, arising from sidelobe-induced correlations. Consequently, two physical states need not become more distinguishable as their separation increases, and the least reliable pair in a codebook need not be its geometrically closest pair. The forbidden-region packing problem must therefore account for all pairwise displacements rather than only local spacing.

We first consider the hexagonal construction developed in Theorem~\ref{thm:hex_lattice_design} using the main-lobe approximation. Fig.~\ref{fig:general}(b) shows the resulting codebook with $J=12$ symbols, which satisfies the reliability condition in Theorem~\ref{thm:codebook_error} when evaluated using the complete field. This indicates that the approximation-based construction can still provide reliable codebooks beyond the local regime. However, reliability is no longer guaranteed by the approximate design criterion and must be verified using the complete field. Indeed, increasing the alphabet to $J=13$ using the same procedure yields a codebook that cannot be certified, as a more distant codeword pair becomes the limiting factor.

We next return to the general forbidden-region packing formulation in \eqref{eq:FJ_def} and \eqref{eq:JF_star_def}, which directly exploits the complete reliability field without imposing a lattice structure. By numerically optimizing the codeword positions, we obtain the nonuniform alphabet shown in Fig.~\ref{fig:general}(c). This codebook contains $J=13$ symbols and satisfies the same reliability requirement, supporting one additional symbol compared with the approximation-based hexagonal construction.

These results highlight a fundamental distinction between local sensing resolution and global communication reliability. A codebook designed solely according to local resolution may overlook nonlocal ambiguities that determine its overall decoding reliability. By translating global sensing distinguishability into explicit packing constraints, the proposed framework characterizes which configurations of physical states can jointly form a reliable communication alphabet. This reveals how the global sensing geometry shapes the structure and size of the embodied alphabet, beyond what can be inferred from local resolution alone.

\section{Conclusion}\label{sec:conclusion}

This paper has positioned embodied communication as a new way to think about information exchange in sensing-rich environments. Its central implication is that communication need not be confined to transmitter-generated waveforms. When agents can intentionally shape the physical states around them, and when infrastructure can perceive those states, the environment itself becomes an information interface. This viewpoint expands the role of future infrastructure from a data pipe or sensing node to an interpreter of purposeful physical change.

The broader impact of embodied communication is that it creates a communication pathway for entities that may not possess, activate, or continuously operate a conventional radio transmitter. Information can instead be conveyed through controllable presence, motion, configuration, reflection, blockage, or other observable physical effects. This may enable low-power, low-complexity, passive, or mechanically constrained agents to participate in networked systems in ways that are difficult to support through conventional wireless links alone.

The vision also goes beyond the cellular and RF setting studied in this paper. Any sensing infrastructure may serve as the receiver, including cameras, LiDAR, acoustic arrays, thermal sensors, magnetic sensors, or multimodal perception systems. Likewise, the communicating agent may be a robot, vehicle, human, wearable object, industrial actuator, reconfigurable surface, or embedded device. Embodied communication should therefore be viewed not as a special case of RF sensing, but as a general architectural principle for sensor-rich intelligent environments.

\bibliographystyle{IEEEtran}
\bibliography{References}

\end{document}